\RequirePackage{ifpdf}
\documentclass[10pt,twoside,a4paper,reqno]{amsart}
\usepackage{amsfonts,amsthm,latexsym,amsmath, amssymb,amscd,epsfig,amstext, mathtools}
\usepackage{textcomp}
\usepackage{wrapfig}
\usepackage{graphics,graphicx}
\usepackage{bbm}
\usepackage[small]{subfigure}

\usepackage{MnSymbol}

\ifpdf
	\usepackage{epstopdf}		
	\DeclareGraphicsExtensions{.png,.jpg,.eps,.epsf}
\fi

\input cyracc.def
\newfam\cyrfam
\font\tencyr=wncyr10
\font\sevencyr=wncyr7

\textfont\cyrfam=\tencyr \scriptfont\cyrfam=\sevencyr
  \scriptscriptfont\cyrfam=\sevencyr

\font\tencyr=wncyr10

\theoremstyle{plain}
\newtheorem{theo+}           {Theorem}
\newtheorem{prop+}           {Proposition}
\newtheorem{coro+}           {Corollary}
\newtheorem{lemm+}           {Lemma}
\newtheorem{conjecture}      {Conjecture}

\theoremstyle{definition}
\newtheorem{defi+}           {Definition}

\newtheorem*{ack}            {Acknowledgement}

\newtheorem{not+}            {Notation}

\theoremstyle{remark}
\newtheorem{rema+}           {Remark}

\newenvironment{theorem}{\begin{theo+}}{\end{theo+}}
\newenvironment{proposition}{\begin{prop+}}{\end{prop+}}
\newenvironment{corollary}{\begin{coro+}}{\end{coro+}}
\newenvironment{lemma}{\begin{lemm+}}{\end{lemm+}}
\newenvironment{remark}{\begin{rema+}}{\end{rema+}}

\newcommand{\al}{\alpha}

\newcommand{\ga}{\gamma}
\newcommand{\Ga}{\Gamma}

\newcommand{\la}{\lambda}

\newcommand{\bC}{\mathbb C}
\newcommand{\bCP}{\mathbb {CP}}
\newcommand{\bR}{\mathbb R}

\newcommand{\HH}{\mathcal H}
\newcommand{\abs}[1]{\mid #1 \mid}

\newcommand{\eps}{\epsilon}
\newcommand {\cP} {\mathfrak P}

\newcommand{\De}{\Delta}
\newcommand{\si}{\sigma}

\newcommand \PP {\mathcal P}

\newcommand{\bB}{\mathcal B}

\newcommand{\ds}{\displaystyle}

\newcommand{\pmtx}[1]{\begin{pmatrix*}[c]#1\end{pmatrix*}}

\catcode`^=\active
\newcommand^[1]{\ensuremath{\sp{{#1}}}}

\numberwithin{equation}{section}

\begin{document}

\title[Level Crossing in Random Matrices. II. ]{Level Crossing in Random Matrices. II. Random perturbation of a random matrix}

\author[T.~Gr\o sfjeld]{Tobias Gr\o sfjeld}
\address{Department of Mathematics, Stockholm University, SE-106 91 Stockholm,
      Sweden}
\email{grosfjeld@math.su.se}

\author[B.~Shapiro]{Boris Shapiro}
\address{Department of Mathematics, Stockholm University, SE-106 91 Stockholm,
      Sweden}
\email{shapiro@math.su.se}

\author [K.~Zarembo]{Konstantin Zarembo} 
\address{AlbaNova Univ. Center, Nordita, Roslagstullsbacken 23, SE-106 91 Stockholm, Sweden}
\email{zarembo@nordita.org}

\date{\today}



\begin{abstract} In this paper we  study   the distribution  of level crossings for the spectra of linear families $A+\la B,$ where  $A$ and $B$ are  square matrices independently  chosen from some given Gaussian ensemble and $\la$ is a complex-valued parameter.  We formulate a number of theoretical and numerical results for the classical Gaussian ensembles and some generalisations. Besides, we present intriguing numerical information about the distribution of monodromy  in case of  linear families for the  classical Gaussian ensembles of $3\times 3$ matrices. \end{abstract}

\maketitle

\section{Introduction}\label{s1}

   Given a linear operator family   
\begin{equation}\label{pencil} 
C=A+\la B,
\end{equation}
 analysis of the dependence  of its spectrum on a perturbative parameter $\la$  is a typical problem both   in fundamental natural sciences and applications, see e.g. the  classical treatise \cite{Ka}.  Depending on the situation $\la$ is considered as   a real or a complex-valued  parameter.

\medskip   
Level crossings of the spectrum  (i.e., collisions of the eigenvalues) in the family \eqref{pencil}  unavoidably occur upon the analytic continuation of a real perturbation parameter  $\la$ into the complex plane, where an intricate pattern of permutations of the eigenvalues arises due to monodromy of the spectrum at each of the level crossing points.
  The positions of  level crossings and monodromy  of the spectrum at each of them constitute an important piece of information about the spectral properties of the linear family \eqref{pencil} and the analytic structure of its spectral surface. Level crossings  determine, in particular, the accuracy of perturbative series in $\la$.  
  
  Since the late 60s, motivated by a number of fascinating observations by C. M. Bender and T. T. Wu \cite{BW}, physicists and mathematicians started considering various cases where $A$ and $B$ are, for example, self-adjoint while $\la$ is complex-valued. A very small sample of such studies can be found in  e.g., \cite{MNOP, Ro, CHM, SH,BDCP, Sm} and references therein.
 
 \medskip 
  Unfortunately, for a somewhat interesting concrete linear family \eqref{pencil}, it is usually quite difficult to exactly describe  the positions of level crossings and especially the monodromy of the spectrum, when $\la$ encircles  closed curves avoiding them.  As an illustration of  specific examples of the physics origin, the reader might consult  \cite{ShTaQu} and \cite{ShT}, where the cases of the quasi-exactly solvable quartic and sextic are considered. The corresponding locations of  level crossings are shown in Figure~\ref{figbla} below.  Although in both cases  numerical experiments reveal  very clear and intriguing patterns for the location of  level crossings as well as the corresponding monodromy,   mathematical proofs explaining these lattice-type patterns  in Figure~\ref{figbla}  are unavailable at present.   
  \begin{figure}
\begin{center}
\includegraphics[scale=0.35]{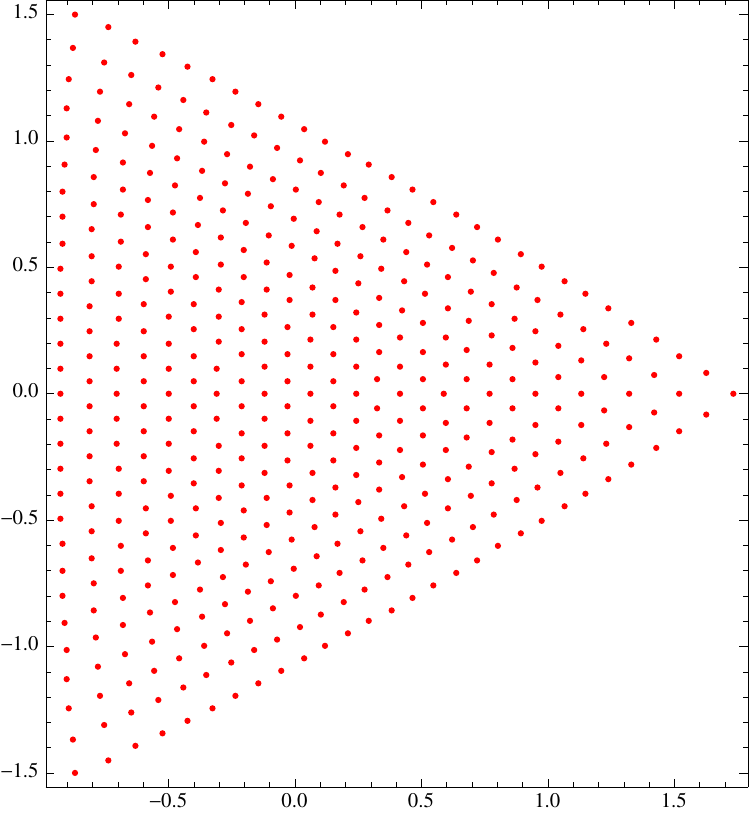} \hskip1cm
 \includegraphics[scale=0.35]{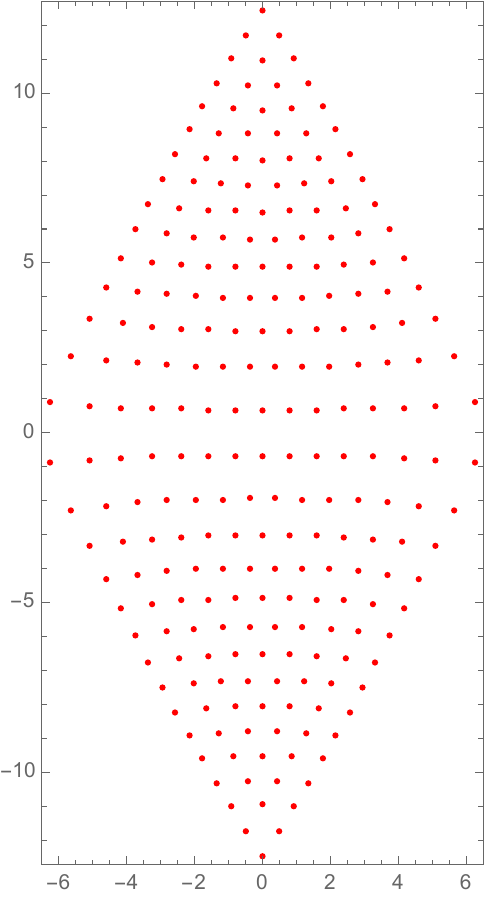} 
\end{center} 

\caption{Level crossings for the quasi-exactly solvable quartic (left) and sextic (right), see \cite{ShTaQu} and \cite{ShT}.}
\label{figbla}
\end{figure}

  \medskip
  Taking this circumstance into account,  in \cite{ShZa1} we considered the problem of  finding the distribution of level crossings within the framework  of the random matrix theory  and studied   the case when $A$ is a fixed matrix while $B$ is a matrix distributed according to one of the standard  Gaussian ensembles.  (To the best of our knowledge, for the first time similar approach  has been used  in  \cite{ZVW}.  For general information on the random matrix theory see e.g. \cite{AGZ}.)

\medskip
  The present paper being a sequel of \cite{ShZa1}, discusses    level crossings   
in linear matrix families of the form \eqref{pencil}, where both  $A$ and $B$ are  independent and equally distributed matrices belonging to a certain class of  complex,  real, real orthogonal or unitary Gaussian ensembles.   To stress the equal r\^ole of matrices in \eqref{pencil}, we denote them here by $A$ and $B$ as opposed to $V_0$ and $H$ in  \cite{ShZa1}. 
  (A somewhat similar situation,  when one  randomly samples  coefficients of a bivariate polynomial instead of the entries of a matrix has been earlier considered in \cite{GP}.)

\medskip
We start with complex Gaussian ensembles. 
Recall that the complex (non-symmetric) Gaussian ensemble $GE_n^\bC$  is the distribution on the space $Mat_n^\bC$ of all complex-valued $n\times n$-matrices,   where each entry of a random $n\times n$-matrix is an independent complex Gaussian variable distributed as $N(0,\frac{1}{2})+i N(0,\frac{1}{2})$.

\medskip 
Our first result is as follows. 

\begin{theorem}\label{th:GE} For any positive integer $n$, if the matrices $A$ and $B$ are independently chosen from $GE_n^\bC,$ then the distribution of  level crossings  in \eqref{pencil} with respect to  the affine coordinate $\la=x+iy$ of $\bC$  is given by 
\begin{equation}\label{densGE}
\mathcal P_{{GE}_n^\bC}(\la):=\mathcal P_{{GE}_n^\bC}(x,y)dxdy=\frac{dxdy}{\pi(1+x^2+y^2)^2}=\frac{dxdy}{\pi(1+|\la|^2)^2}.
\end{equation}
\end{theorem}

\medskip
\begin{remark} In polar coordinates $(r,\theta)$ in the complex plane of parameter $\la$, the above distribution $\mathcal P_{{GE}_n^\bC}(\la)$ has the form
$$ 
\mathcal P_{{GE}_n^\bC}(r,\theta)drd\psi=\frac{rdrd\theta}{\pi(1+r^2)^2},
$$
giving the radial CDF of the form 
$$\Psi_{{GE}_n^\bC}(r)=\frac{r^2}{1+r^2}.$$
\end{remark}

\medskip
\begin{remark} \label{rmk2}
Let us realize $\bCP^1 \simeq  S^2$ as the unit sphere in $\bR^3$ with coordinates $(X,Y,Z)$ and identify  the complex plane of parameter $\la=x+iy$ with the horizontal coordinate $(X,Y)$-plane, where $X$ corresponds to the real axis  and $Y$ corresponds to the imaginary axis in $\bC$. If we  use the standard stereographic projection of the unit sphere in $\bR^3$ from its north pole, i.e., from the point $(0,0,1)$ onto the $(X,Y)$-plane, then the usual area element of the sphere induced from the standard Euclidean structure in $\bR^3$ is given by
$$dA=\frac{4dxdy}{(1+x^2+y^2)^2}=\frac{4dxdy}{(1+|\la|^2)^2}.$$

The latter fact implies that the r.h.s. of \eqref{densGE} presents the constant density $\frac{1}{4\pi}$  with respect to the standard Euclidean area measure on $S^2\simeq \bC P^1$  compactifying the complex plane of parameter $\la$.  (The constant density  $\frac{1}{4\pi}$ provides the unit sphere with the total mass $1$.) 
 \end{remark}
 
 \begin{remark} 
 Observe that    formula \eqref{densGE} is independent of the size of  $A$ and $B$ (and  also of the variance of the matrix ensemble, if we allow to change it). For $n=1$,  formula \eqref{densGE} gives the distribution of the quotient of two independent  complex Gaussian random variables.  
\end{remark}

\begin{figure}

\begin{center}
 \includegraphics[scale=0.4]{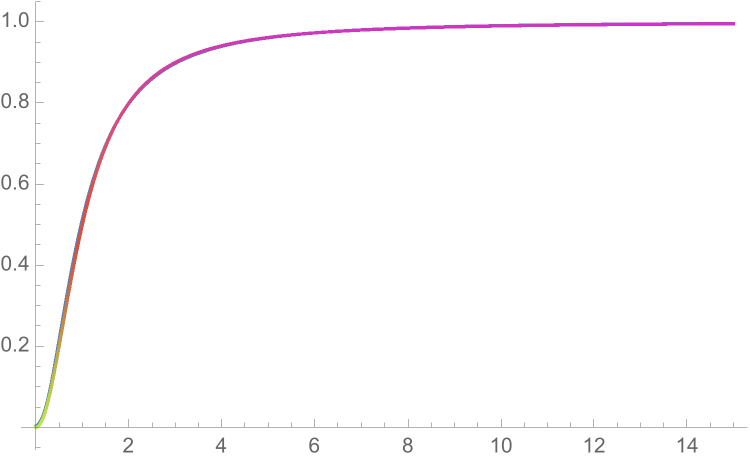} 
\end{center}

\caption{Radial density of  level crossings for $A+\la B,$ where $A$ and $B$ are independently sampled from $GE_{6}^\bC$; (we used  100 random pairs). The above diagram shows a perfect match of the numerical distribution of the absolute values of  level crossings obtained in our sampling  with the theoretical radial CDF  $\frac{r^2}{1+r^2}$.}
\label{fig1}
\end{figure}

A number of further generalizations of  Theorem~\ref{th:GE} can be found in  \S~\ref{sec2}.

\medskip
Next we consider    Gaussian orthogonal,  Gaussian unitary, and real Gaussian  ensembles.  
Recall that 

\smallskip
\noindent
 (i) the  Gaussian orthogonal ensemble  $GOE_n^\bR$  is the distribution on the space  $Sym_n^\bR$ of real-valued symmetric matrices,  where each entry $e_{i,j}=e_{j,i},\; i < j$ of a matrix is an independent random variable  distributed as  $N(0,1)$, and each diagonal entry $e_{i,i}$ is independently distributed as $\sqrt{2} N(0,1)$;  

\smallskip
\noindent
 (ii) the Gaussian unitary ensemble  $GUE_n$-ensemble  is the distribution on the space $\HH_n$ of all Hermitian $n\times n$-matrices, where each entry $e_{i,j}=e_{j,i},\; i < j$ of a matrix is an independent random variable  distributed as  $N(0,\frac{1}{2})+i N(0,\frac{1}{2})$, and each diagonal entry $e_{i,i}$ is independently distributed as $N(0,1)$;

\smallskip
\noindent 
(iii) the real (non-symmetric) Gaussian ensemble  $GE_n^\bR$   is the distribution on the space $Mat_n^\bR$ of real-valued $n\times n$ matrices,  where each entry of a matrix is an independent real random variable distributed as $N(0,1)$.  

\medskip
In the case of $GOE_n^\bR$ we have a theoretical result for $n=2$ and  a conjecture for $n\ge 3$ based on computer simulations. 

\begin{theorem}\label{th:GOE} If the matrices $A$ and $B$ are independently chosen  from $GOE_2^\bR$,  then the distribution of  level crossings in \eqref{pencil} is uniform on $\bCP^1 \supset\bC,$ i.e., their density is given by the right-hand side of \eqref{densGE}.   
\end{theorem}

\begin{remark} One can easily check that the distribution of  level crossings for $A$ and $B$ independently taken from $GOE_1^\bR$ is uniform on the real projective line $\bR P^1$. 
\end{remark}

Extensive numerical experiments strongly support the following guess illustrated in Fig.~\ref{figGO}. 
\begin{conjecture}\label{conj:GOE}  For any fixed size $n>2$, if the matrices $A$ and $B$ are independently chosen  from $GOE_n^\bR$,  then the distribution of  level crossings in \eqref{pencil} is uniform on $\bCP^1 \supset\bC.$
\end{conjecture}

\begin{remark} Notice that on Fig.~\ref{figGO}  one can hardly see the difference  between the statistical results for $n=2,4,6,8,10$ and the theoretical CDFs of the uniform distribution on $\bCP^1$.  
 Although the simple (conjectural) answer for  level crossing distribution in the $GOE$-case presented in Theorem~\ref{th:GOE} and Conjecture~\ref{conj:GOE}  indicates the possible existence of some extra symmetry complementing the  $SO_2$-action presented in \S~\ref{sec3}, we were not able to find such.
\end{remark}



 \begin{figure}
\begin{center}
\includegraphics[scale=0.3]{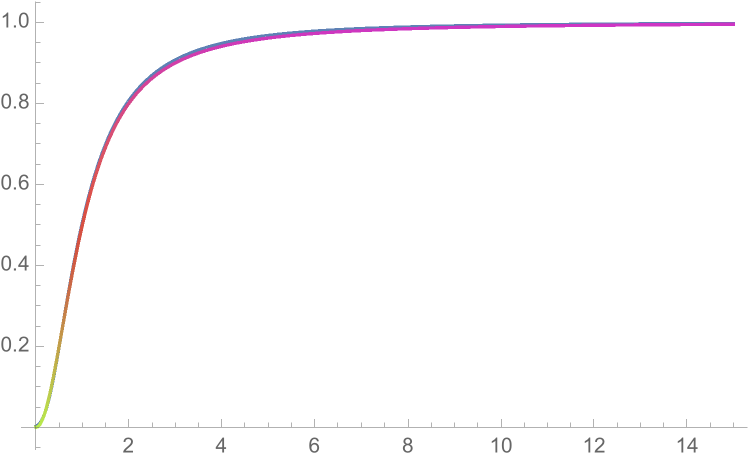} \hskip 2cm \includegraphics[scale=0.3]{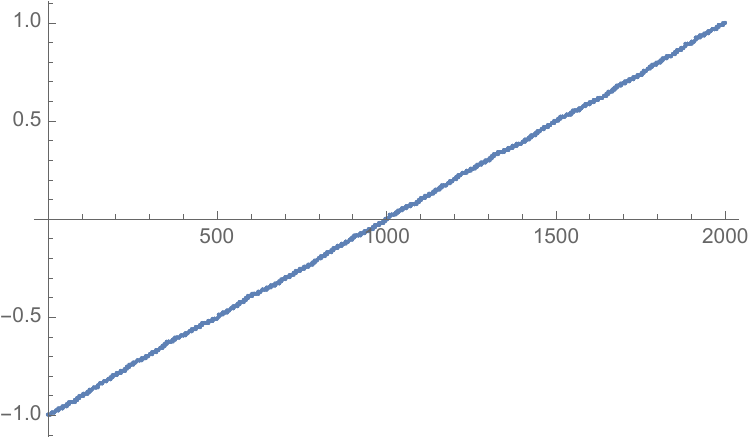}   

\end{center}
\caption{Numerical and theoretical radial and angle CDFs for \eqref{pencil} for $n=2,4,6,8,10$  with $A$ and $B$  taken from $GOE_n$ are practically indistinguishable.}
\label{figGO}
\end{figure}

\medskip
Our next results deal with   Gaussian unitary ensembles. Here again we have a theoretical result for $n=2$ and numerical plots for higher $n$.

\begin{theorem}\label{th:GUE}  If the matrices $A$ and $B$ are independently chosen from $GUE_2,$ then the distribution of   level crossings in $\bC$ is given by 

\begin{equation}\label{densGUE2}
\mathcal P_{GUE_2}(x,y)dxdy=  \frac{4|y|dxdy}{\pi(1+x^2+y^2)^3}=\frac{1}{\pi}\left \vert \frac{y}{1+x^2+y^2}\right \vert \frac{4dxdy}{(1+x^2+y^2)^2},
\end{equation}
which matches the general formula~\eqref{eq:form}. 



\end{theorem}

In the cylindrical coordinates $(\psi,Y)$ on $\bC P^1$, where $0\le \psi \le 2\pi$ and $-1\le Y \le 1$,  one has 
\begin{equation}\label{eq:GUE2}
\mathcal P_{GUE_2}(\psi,Y) d \psi dY=\frac{|Y| d\psi dY}{2\pi}.
\end{equation}






 
 \begin{figure}
\begin{center}
\includegraphics[scale=0.65]{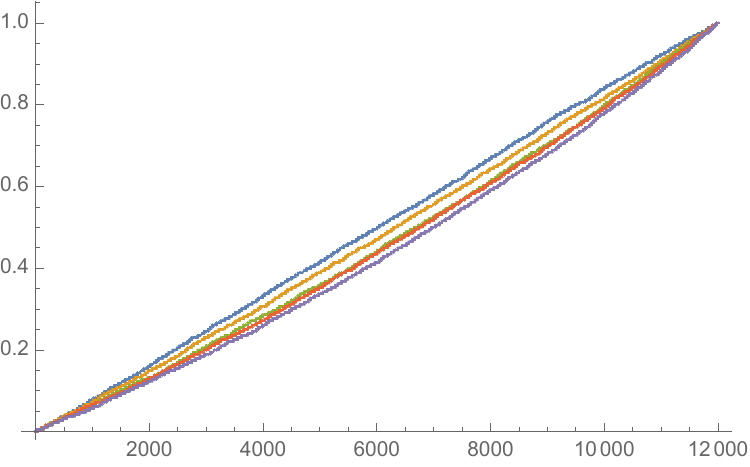} 
\end{center}
\caption{Empirical distributions of  $|Y|$ for \eqref{pencil}�  taken from $GUE_n$ with $n=2,3,4,5,6$. (Curves corresponding to the increasing values of $n$  lie one below the other; the blue straight line corresponds to $n=2$, see \eqref{eq:GUE2}.)}
\label{fig23}
\end{figure}





 
 \medskip
At present, we do not have  explicit (even conjectural) formulas for the densities $\mathcal P_{{GUE}_n}(x,y)$,  for $n\ge 3$ similar to  \eqref{eq:GUE2}. But  we carried out substantial numerical experiments for matrix sizes up to $6$ conducted as follows. For each $n \in \{2, . . . , 6\},$  sampling  independently pairs of $GUE_n$-matrices, we calculated 12,000 level crossing points for every $n$ 
  and plotted the  values of $|Y|$ for   obtained level crossings in
increasing order, see Fig.~\ref{fig23}.  These numerical experiments  strongly suggest the following.

\begin{conjecture}\label{conj:main} There exists a limiting distribution $ \mathcal P_{GUE_\infty}(Y):=\lim_{n\to \infty} \mathcal P_{GUE_n}(Y)$.
\end{conjecture}


\medskip
Our final results deal with the case when $A$ and $B$ are independently taken from the  $GE_n^\bR$-ensemble. Theoretical results are available for $n=2$ as well as  an explicit general conjecture about the asymptotics of level crossings when $n\to \infty$.   
The next statement describes the distribution of the coefficients of the random real quadratic discriminantal polynomial whose roots are the two level crossing points $(\la_+, \la_-)$ in the  situation when $A$ and $B$ independently taken from the  $GE_2^\bR$-ensemble.

\medskip
\begin{proposition}\label{prop:GE2Rparts} \rm{(i)} The density of the average of the two level crossing points $(\la_+, \la_-)$  with respect to the Lebesgue measure on the real axis is given by  the following single integral
\begin{align}
\rho_{\frac{\lambda_+ + \lambda_-}{2}}(x) 
&= \int_{-1}^{1} \frac{\left| t\right|}{\pi\sqrt{2-2 t} \left(x ^2 t^2+1\right)^2} dt, 
\end{align}
where $x\in \bR$.

\medskip\noindent
\rm{(ii)} 
 The density of the product of the two level crossing points $(\la_+, \la_-)$  with respect to the Lebesgue measure on the real axis is given by 
\begin{align}
\begin{split}
\rho_{\frac{D_A}{D_B}}(x) &= \Theta(x)\left[\frac{1}{2(x+1)^2} -\int_{-\infty}^0 \frac{y e^{-y(1+x)/2}}{8}\text{erfc}(\sqrt{-y})\text{erfc}(\sqrt{-xy})dy\right] \\
 &+ \Theta(-x) \left[ \frac{1}{(x+1)^2}\left(1+\frac{3x-1 + (x-3) \sqrt{-x}}{\sqrt{8} (1-x)^{3/2}}\right)\right],
\end{split}
\end{align}
where $x\in \bR$ and $\text{erfc}(t)$ stands for the standard complementary error function given by $\text{erfc}(t)=\frac{2}{\sqrt{\pi}}\int_t^\infty e^{-\tau^2}d\tau.$
\end{proposition}

For the actual distribution of the level crossings on the complex $\la$-plane we were only able to obtain the following complicated claim.   

\begin{proposition}\label{prop:GE2R}
\rm{(i)} For $\la=x+iy$ and $y\neq 0$, the distribution of level crossings of \eqref{pencil} with $A$ and $B$ independently taken from the  $GE_2^\bR$-ensemble is given by the  triple integral: 
 \begin{align}
\begin{split} \PP_{GE_2^\bR}(x,y)dxdy = \int_{-\infty}^\infty da \int_0^\infty dr \int_{-\infty}^\infty db \cdot e^{-\frac{r^2+ b^2+((r^2 - b^2)(x^2 + y^2) + \left(\frac{ar + x(r^2 - b^2)}{b}\right)^2 }{2}}\\
 \cdot \left|\frac{yr}{2\pi^2 b}(r^2 - b^2)^2\right| \cdot \frac{ \Theta\left({(r^2 - b^2)(x^2 + y^2) + \left(\frac{a r + x(r^2 - b^2)}{b}\right)^2 - a^2}\right)}{\sqrt{(r^2 - b^2)(x^2 + y^2) + \left(\frac{a r + x(r^2 - b^2)}{b}\right)^2 - a^2}}dxdy,
 \end{split}
 \end{align}
where $\Theta$ is the Heaviside $\Theta$-function, i.e. $\Theta (t)=0$ for $t<0$ and $\Theta (t)=1$ for $t>0$.

\medskip\noindent
\rm{(ii)} 
\begin{equation}\label{eq:realaxis}
\PP_{GE_2^\bR}(x,0)dxdy=\frac{\sqrt{2}}{\pi}\frac{dx\delta y}{(1+x^2)^2}.
\end{equation}

\end{proposition}

\medskip
It seems really difficult to  get any explicit formulas for the distributions of  level crossings of $GE_n^\bR$ with $n\ge 3$,  
  but  as in the previous cases, we performed detailed numerical experiments illustrated in Fig.~\ref{fig3} and \ref{fig3e}. These experiments strongly suggest the validity of the following guess to which we plan to return in  \cite{GrShZa3}.  

\begin{conjecture}\label{conj:imp} 
 When $n\to \infty$, the  level crossing distribution for $A$ and $B$ independently sampled from $GE_n^\bR$ 
approaches  the uniform distribution on $\bCP^1$.
\end{conjecture}

\begin{figure}

\begin{center}

\includegraphics[scale=0.07]{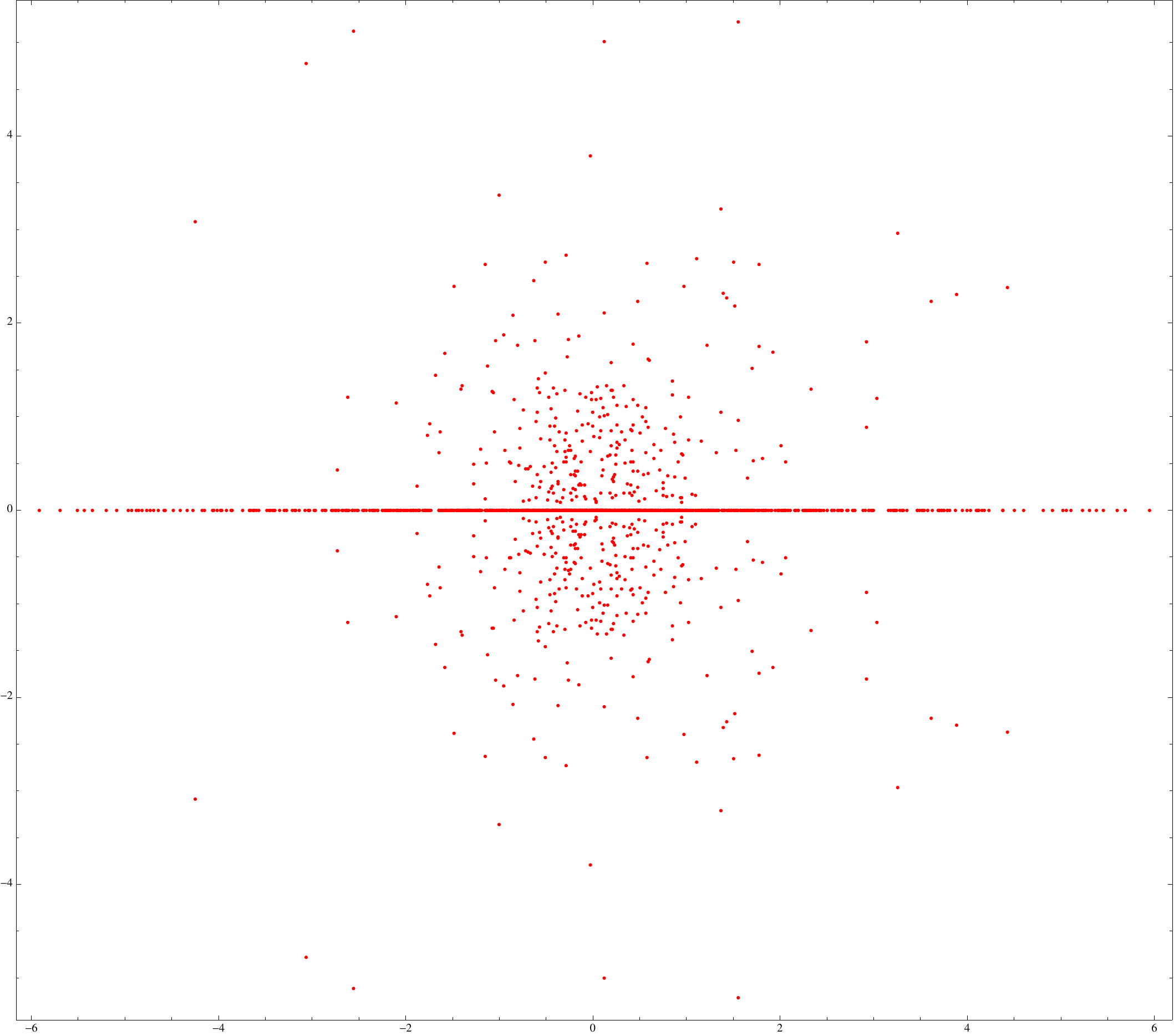}\hskip1cm \includegraphics[scale=0.07]{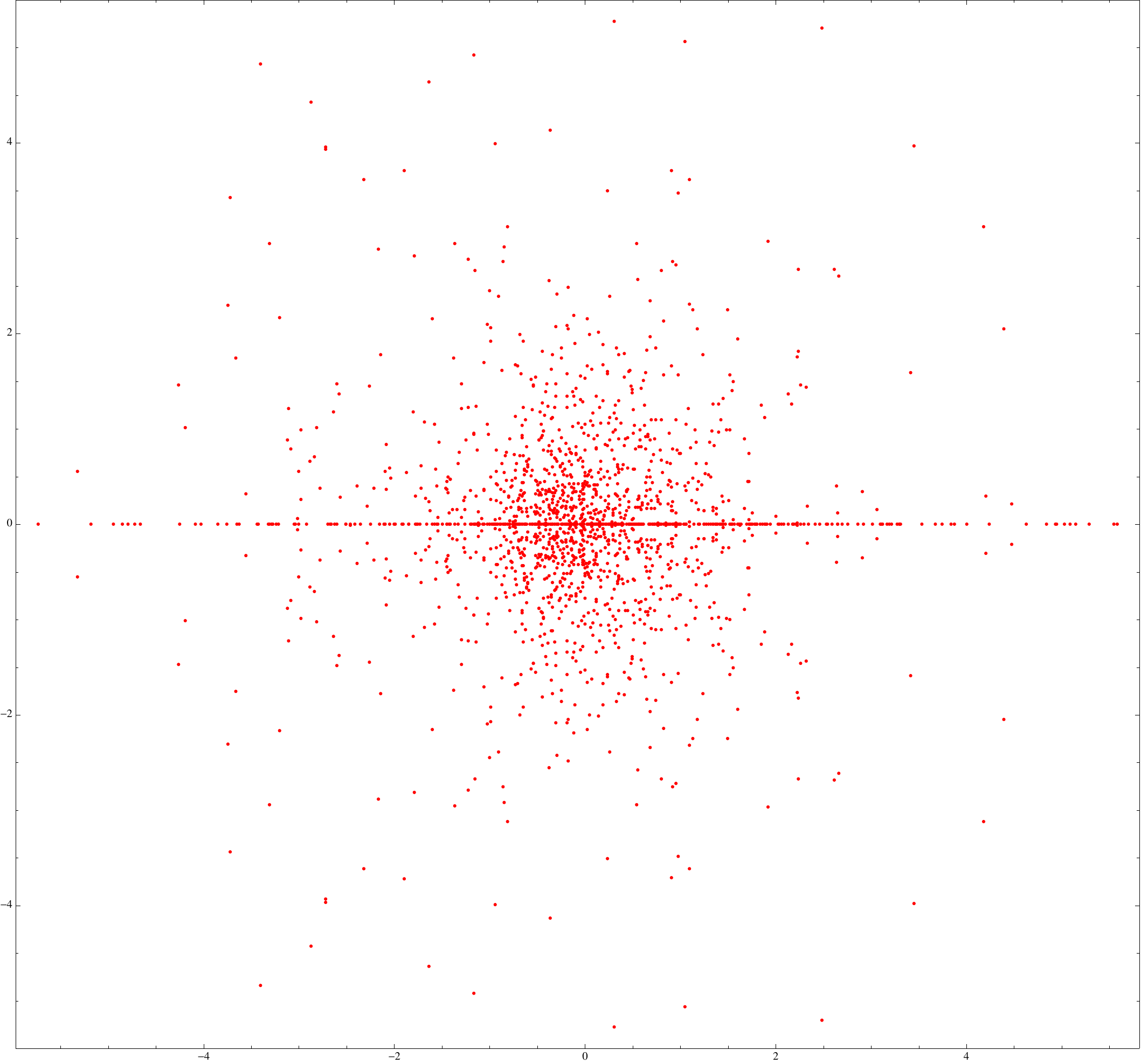}\hskip1cm \includegraphics[scale=0.065]{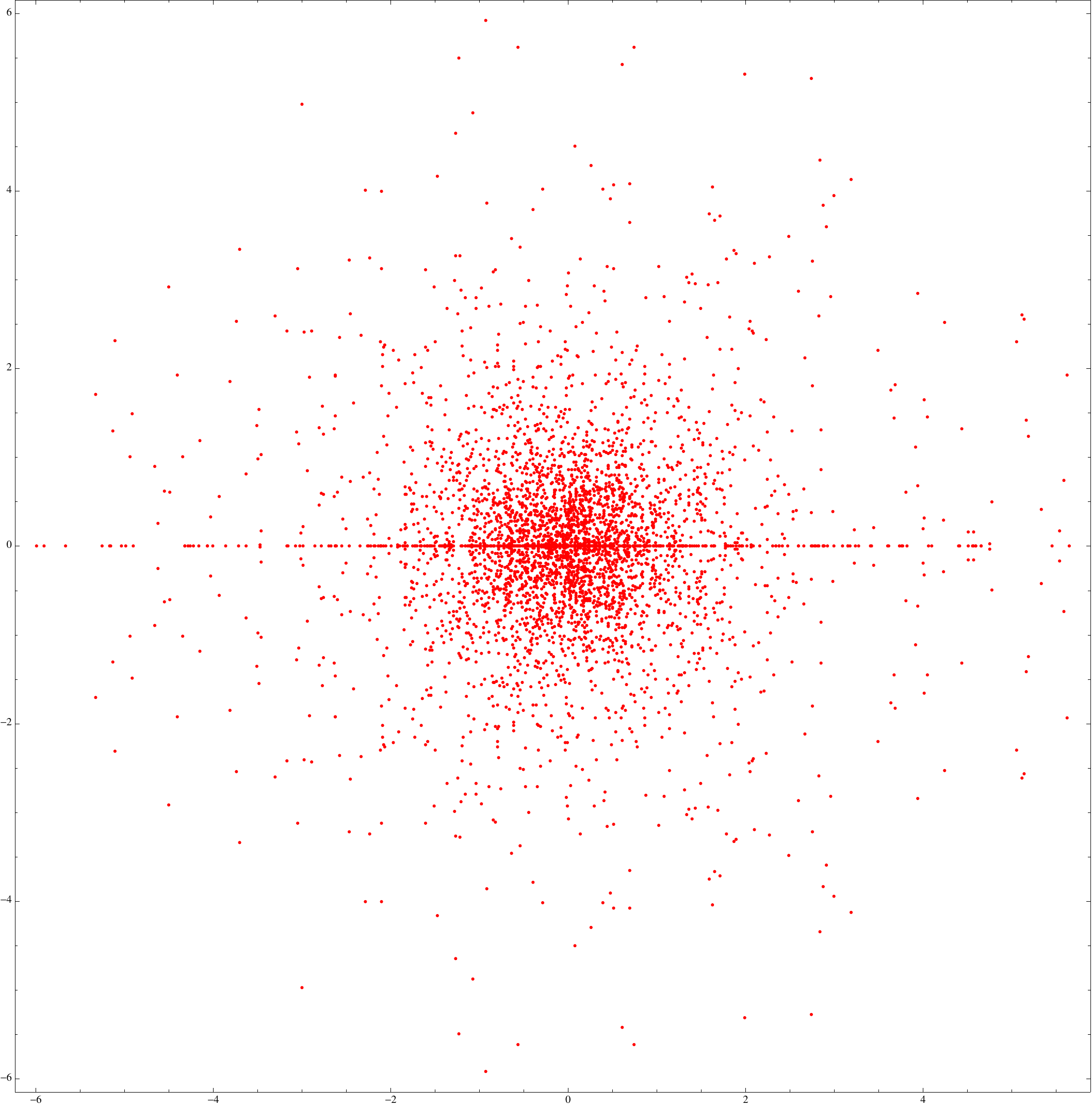}

\end{center}


\caption{Distributions of   level crossings in the $\la$-plane when $A$ and $B$ are sampled from $GE^\bR_n$ for $n=2, 5, 10$ apparently approaching  the uniform distribution on $\bC P^1$.}
\label{fig3}
\end{figure}

\bigskip
\begin{figure}\label{FigFig}

\begin{center}

\includegraphics[scale=0.25]{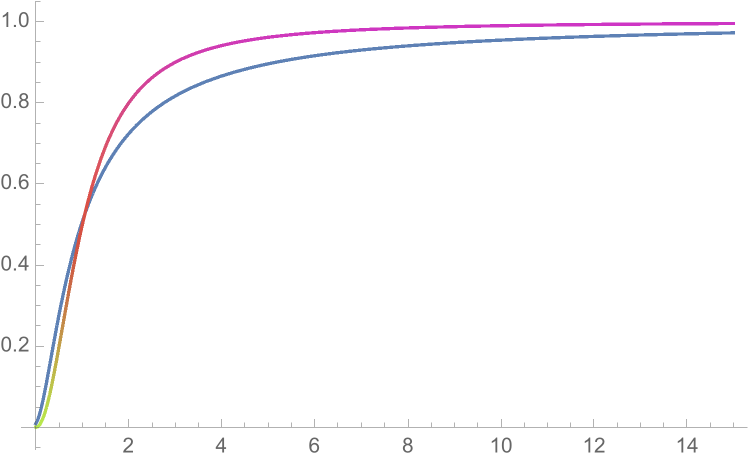}\hskip2cm\includegraphics[scale=0.25]{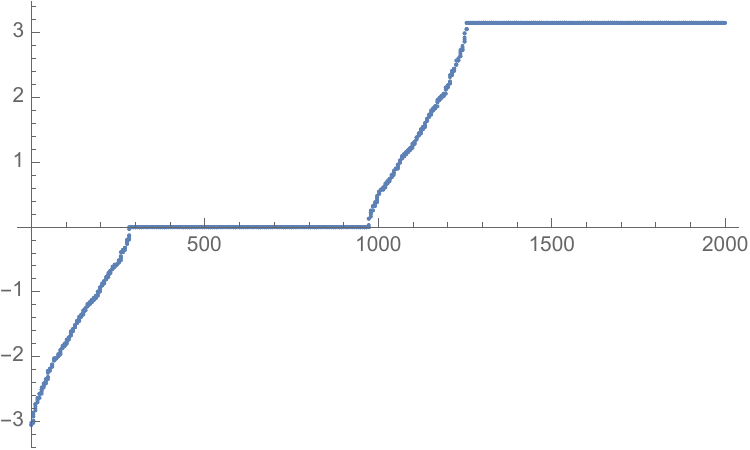} 
\vskip1cm 
\includegraphics[scale=0.25]{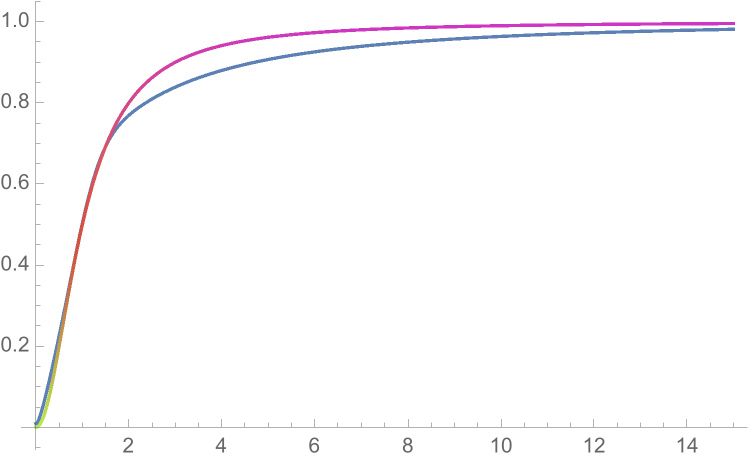}\hskip2cm\includegraphics[scale=0.25]{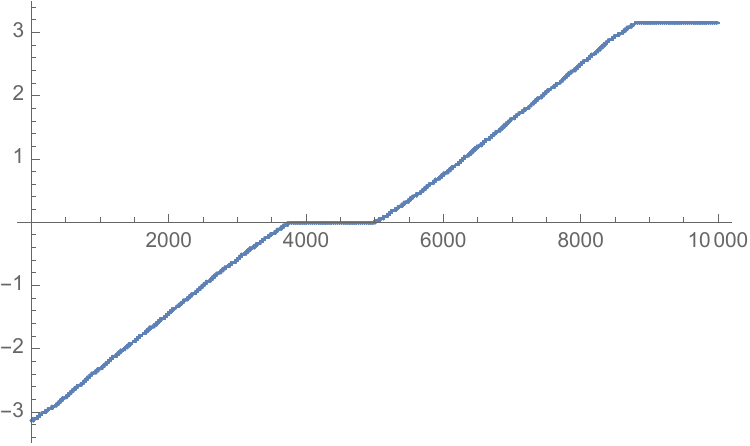} 
\vskip1cm
\includegraphics [scale=0.25]{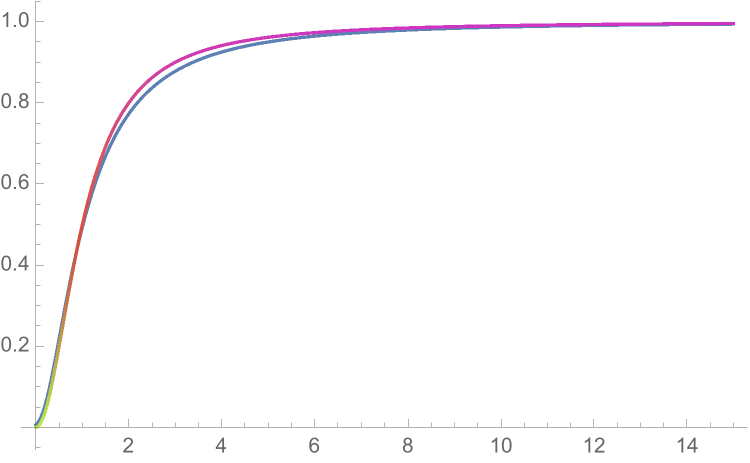}\hskip2cm\includegraphics[scale=0.25]{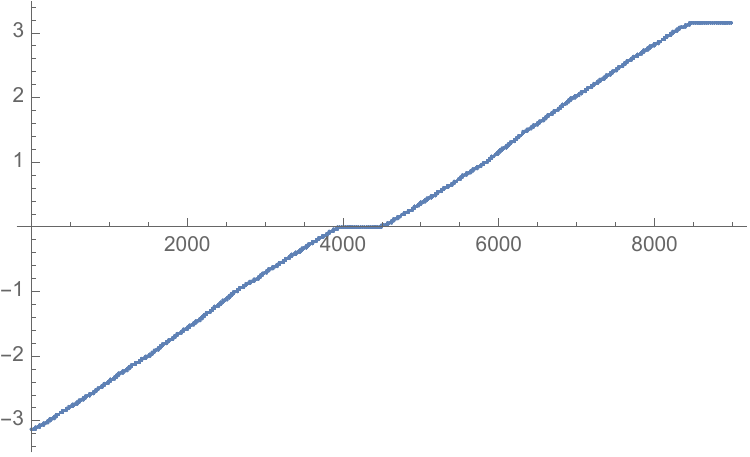}

\end{center}


\caption{Radial and angle distributions of  level crossings with $A$ and $B$ sampled from $GE^\bR_n$ with $n=2, 5, 10$ approaching that of  the uniform distribution on $\bC P^1$. (The limiting theoretical radial density is shown by the magenta line and the experimental results by the blue line.)}
\label{fig3e}
\end{figure}

We have also numerically evaluated the number of real level crossings among the total number of level crossings. (Real level crossing are represented by the horizontal segments of the graphs in the right column of Fig.~\ref{fig3e}.) Our numerics suggests that for a given size $n$, the average number of real level crossings is close to $\sqrt{n(n-1)}$ which is the square root of the total number of level crossings (given by $n(n-1)$). (Observe that  in many similar situations involving real random univariate polynomials it is known that the average  number of real roots equals the square root of their degree. Unfortunately, our situation is not covered by the known theoretical results.) We can prove that $\sqrt{2}$ is the expected average for $n=2$, see Lemma~\ref{lm:sqrt2}. For $n=3,4,5$ with 10000 samples, the quotient of the empirical average divided by $\sqrt{n(n-1)}$ was $1.0405, 1.0404, 1.04957$ resp. For $n=6$ with 5000 samples, the same quotient was  $1.05586$ and, finally, for $n=10$ with 130 samples, the quotient was  $1.06382$.


\medskip
The structure of the paper is as follows. In \S~\ref{sec2} we prove some general introductory  results and make conclusions  about the complex Gaussian ensembles. In \S~\ref{sec3}, we discuss the $SO_2$-action on $\bC P^1$ and Gaussian ensembles.  In \S~\ref{sec4}, we consider the cases of orthogonal Gaussian ensembles and Gaussian unitary ensembles and  settle Theorems~\ref{th:GOE} and ~\ref{th:GUE}. In \S~\ref{sec55}, we  settle Propositions~\ref{prop:GE2Rparts}  and \ref{prop:GE2R} for the real Gaussian ensemble. Finally,  in \S~\ref{sec6}, we  present interesting numerical results about the monodromy  statistics of $3\times 3$ linear families \eqref{pencil}.

\medskip 
\begin{ack} B.S. wants to thank the Department of Mathematics of UIUC and, especially,  Professor Y.~Baryshnikov for the hospitality in June-July 2015 when a part of this project was carried out. 
The research of  K.Z. was supported by the Marie 
Curie network GATIS of the European Union's FP7 Programme under REA Grant
Agreement No 317089, by the ERC advanced grant No 341222, by the Swedish Research Council (VR) grant
2013-4329, and by RFBR grant 18-01-00460.  
\end{ack}

\section{$SU_2$-action and complex Gaussian  ensembles}\label{sec2}

To prove our results about complex Gaussian ensembles, we  need the following construction.  The $GE_n^\bC$-probability measure $\ga:=\ga_{GE_n^\bC}$ on $Mat_n^\bC$ induces the product probability measure  $\ga^{(2)}$ on $Mat_n^\bC\times Mat_n^\bC$. Consider the {\em spectral determinant} $\mathfrak D_n\subset Mat_n^\bC\times Mat_n^\bC\times \bC,$  which is a complex algebraic hypersurface consisting of all triples $(A,B,\la)$ such that the matrix $A+\la B$ has a multiple eigenvalue. 
Projection $\pi_n: \mathfrak D_n\to Mat_n^\bC\times Mat_n^\bC$ by forgetting the last coordinate induces a branched covering of $Mat_n^\bC\times Mat_n^\bC$ by $\mathfrak D_n$ of degree $n(n-1)$ whose fiber over a pair $(A,B)$ coincides with  level crossing set  of the linear family $A+\la B$. Taking the pullback $\pi^{-1}_n(\ga^{(2)}),$ we obtain the probability measure $\Ga:=\Ga_{GE_n^\bC}$ on $\mathfrak D_n$.  (In other words,  for any open subset $\mathcal O\subset \mathcal D_n$ which projects diffeomorphically on its image, $\Ga(\mathcal O)=\frac{1}{n(n-1)}\ga(\pi_n(\mathcal O))$. Similar construction can be used for any branched covering whose base is equipped with an arbitrary probability measure.)  

Now let $\kappa_n: \mathfrak D_N \to \bC$ be the projection of the spectral determinant onto the last coordinate in $Mat_n^\bC\times Mat_n^\bC\times \bC$, i.e., onto the $\la$-plane. Then the measure $\mu:=\mu_{GE_n^\bC}$ we are looking for, coincides with  the pushforward $\mu:=\kappa_n(\pi^{-1}_n(\ga^{(2)}))$.  (In other words, the value of measure $\mu$ on any measurable subset of $\bC$ equals the value of measure $\Ga$ of its complete preimage in $\mathcal D_n$.) 

\smallskip
For our purposes, it will be more convenient to consider the space $Mat_n^\bC\times Mat_n^\bC\times \bCP^1,$ with the inclusion $\bC\subset \bCP^1$ given by the stereographic projection introduced in Remark~\ref{rmk2}.  In other words, we use $\la:=b/a$, $(a:b)$ being the homogeneous coordinates on $\bCP^1$. The above constructions work equally well on $Mat_n^\bC\times Mat_n^\bC\times \bCP^1$ and provide us with  the measure $\mu$ supported on $\bCP^1$.   (By a slight abuse of notation we denote both measures by the same letter.) 

\smallskip  \medskip
  Consider the following  $SU_2$-action on   $Mat_n^\bC\times Mat_n^\bC$. A matrix  $\frak U \in SU_2$ given by $\begin{pmatrix} u & -\bar v \\ v&\bar u \end{pmatrix},\;$ $|u|^2+|v|^2=1$ acts on the latter product space  by: 
\begin{equation}\label{eq:rest}
(A,B) \ast \frak U \mapsto (u  A+v  B, -\bar v A+ \bar u  B).
\end{equation} 

Consider the following  $SU_2$-action on   $Mat_n^\bC\times Mat_n^\bC\times \bCP^1$ extending the above action \eqref{eq:rest}. 
\newline
A matrix $\mathfrak U= \begin{pmatrix} u & -\bar v \\ v&\bar u \end{pmatrix},\;$ $|u|^2+|v|^2=1$  acts on $Mat_n^\bC\times Mat_n^\bC\times \bCP^1$ by: 
\begin{equation}\label{eq:action}
 (A,B,a:b) \ast \mathfrak U \; \mapsto (u  A+v  B, -\bar v A+ \bar u  B, \bar u  a +\bar v b: -v a +u b).
\end{equation}
  Observe that the third component of the latter action coincides with  the standard $SU_2$-action   on a point $(a:b)\in \bCP^1$ of the conjugate matrix $\begin{pmatrix} \bar u & - v \\ \bar v& u \end{pmatrix}\;$.   

\medskip
To prove Theorem~\ref{th:GE} stated in the Introduction, we will  show that $\mu$ is invariant under the above $SU_2$-action on $\bCP^1.$ Since this action preserves the standard Fubini-Study metric on $\bC P^1,$ we can conclude  that its density is constant with respect to the area form induced by the Fubini-Study metric, i.e.,  the one which has constant density in the cylindrical coordinates $(\phi, Z)$. 

\smallskip
Our proof of Theorem~\ref{th:GE}  consists of three steps.  On step 1 we will show that the action \eqref{eq:action}  on $Mat_n^\bC\times Mat_n^\bC\times \bCP^1 $ preserves the spectral determinant   
$\widehat {\mathfrak D}_n\subset Mat_n^\bC\times Mat_n^\bC\times \bCP^1$. On step 2 we will prove that this action preserves the probability measure $\ga^{(2)}$ on $Mat_n^\bC\times Mat_n^\bC$. As a consequence of steps 1 and 2, it also preserves the probability measure $\pi_n^{-1}(\ga^{(2)})$ on $\widehat {\mathfrak D}_n$.  On step 3 we will show the equivariance of \eqref{eq:action} with respect to the projections $\pi_n$ and $\kappa_n$.

\begin{lemma}\label{lm:preserv}
The action \eqref{eq:action} preserves $\widehat {\mathfrak D}_n\subset Mat_n^\bC\times Mat_n^\bC\times \bCP^1$.
\end{lemma}

\begin{proof}
Take an arbitrary triple $(A,B,a:b)$ belonging to $\widehat {\mathfrak D}_n$, i.e., such that $aA+bB$ has a multiple eigenvalue, and take any $\mathfrak U= \begin{pmatrix} u & -\bar v \\ v&\bar u \end{pmatrix}\in SU_2$.  We need to show that the triple $$(u  A+v  B, -\bar vA+ \bar u  B, \bar u  a +\bar v b: -v a +u b)$$ also belongs to 
$\widehat {\mathfrak D}_n$. In other words, we need to check  that  if  $aA+bB$ has a multiple eigenvalue, then the matrix 
$$( \bar u  a +\bar v b)(u  A+v  B)+(-v a +u b)(-\bar v A+ \bar u  B)$$ 
has a multiple eigenvalue as well. The latter claim is obvious since expanding the above expression, we get $aA+bB$.   
\end{proof}

\begin{proof}[Proof of Theorem~\ref{th:GE}] 
To settle step 2, observe that in case of the  $GE_n^\bC$-ensemble, the probability density  to obtain a  matrix $A\in Mat_n^\bC$ is given by: 
$$\ga(A)=\frac{1}{\pi^{n^2}}e^{-\sum_{i,j=1}^n|A_{ij}|^2}=\frac{1}{\pi^{n^2}}e^{-Tr(AA^\star)},$$
where $A^\star$ stands for the conjugate-transpose of $A$. Therefore the density of $\ga^{(2)}$ on $Mat_n^\bC\times Mat_n^\bC$ is given by: 
 $$\ga^{(2)}(A,B)=\frac{1}{\pi^{2n^2}}e^{-Tr(AA^\star +BB^\star)}.$$
Setting $C=uA+vB$ and $D=-\bar v  A+\bar u B$, we get the relation
 $$Tr(CC^\star +DD^\star)=Tr(AA^\star+BB^\star)=$$
$$=Tr(u\bar u A A^\star+v\bar u B A^\star +u\bar v A B^\star+v\bar v B B^\star  +v\bar v A A^\star-\bar u v B A^\star  -u\bar v A B^\star+u\bar u B B^\star ).$$
The latter equality implies that the action \eqref{eq:action} restricted to $Mat_n^\bC\times Mat_n^\bC$ (i.e., forgetting its action on the last coordinate $\bC P^1$) preserves $\ga^{(2)}$. 
By Lemma~\ref{lm:preserv},  the action \eqref{eq:action} preserves the hypersurface $ \widehat {\mathfrak D}_n$ and, therefore it preserves the probability measure $\pi_n^{-1}(\ga^{(2)})$ on it. 

To settle step 3, we need to show that the measure $\mu:=\kappa_n(\pi^{-1}_n(\ga^{(2)}))$ on $\bC P^1$ is invariant under the conjugate action of $SU_2$ on $\bC P^1$, see the last component of  \eqref{eq:action}. Take an arbitrary open set $\Omega\subset \bC P^1$ and $g\in SU_2$. Denote by $g\cdot \Omega \subset \bC P^1$ the shift of $\Omega$ by the conjugate of $g$. We need to prove that $\mu(\Omega)=\mu(g\cdot \Omega)$. By definition, $\mu(\Omega):=\pi^{-1}(\ga^{(2)})(\kappa_n^{-1}(\Omega))$ and $\mu(g\cdot \Omega):=\pi^{-1}(\ga^{(2)})(\kappa_n^{-1}(g \cdot\Omega))$. (Observe that both $\kappa_n^{-1}(\Omega)$ and $\kappa_n^{-1}(g\cdot \Omega)$ are measurable subsets of $\widehat {\mathfrak D}_n$.)  Let us show that the \eqref{eq:action}-action  by $g$ sends $\kappa_n^{-1}(\Omega)$ to  $\kappa_n^{-1}(g\cdot \Omega)$ and the  \eqref{eq:action}-action  by the inverse $g^{-1}$ sends $\kappa_n^{-1}(g\cdot \Omega)$ to  $\kappa_n^{-1}(\Omega)$ implying the required coincidence of measures due to  step 2. Indeed $\kappa_n^{-1}(\Omega)$ is the set of all triples 
$(A,B, a:b)$ such that $aA+bB$ has a multiple eigenvalue and $(a:b)\in \Omega$.  By Lemma~\ref{lm:preserv}, acting by $g$ on any such triple we get another 
triple $(\tilde A, \tilde B, \tilde a:\tilde b)$ such that $\tilde a\tilde A+\tilde b\tilde B$ has a multiple eigenvalue and $(\tilde a:\tilde b)\in g\cdot \Omega$. The same argument applies to  the \eqref{eq:action}-action  by the inverse $g^{-1}$. 
\end{proof}

\medskip

\begin{remark} 
\medskip
Observe that  an alternative way to express the fact that the r.h.s. of \eqref{densGE} presents the constant density $\frac{1}{4\pi}$  with respect to the standard Euclidean area measure on $S^2\simeq \bC P^1$  is as follows. Consider the  standard cylindrical coordinate system $(\rho, \phi, Z)$ in $\bR^3,$ where $\rho \ge 0, 0\le \phi \le 2\pi, Z\in \bR$.  Recall that  
$$X=\rho \cos \phi,\; Y=\rho \sin \phi,\; Z=Z.$$ 
 If we consider $(\phi,Z)$,  $0\le \phi \le 2\pi, -1\le Z \le 1,$ as coordinates on the unit sphere $S^2\simeq \bC P^1$ (with both poles removed), then in these coordinates the usual area element  on the sphere is given by 
$$dA=d\phi dZ.$$ Thus, in cylindrical coordinates $(\phi, Z)$, $0\le \phi \le 2 \pi;\; -1\le z \le 1$ parameterising  the unit sphere $S^2$, the measure $\mathcal P_{{GE}_n^\bC}(x,y)dxdy$ given by  \eqref{densGE} transforms into   
\begin{equation}\label{eq:const}
\mathcal P_{{GE}_n^\bC}(\phi,Z)d\phi dZ=\frac{d\phi dZ}{4\pi}.
\end{equation} 

In the case of $2\times 2$-matrices, the formula
$$
 \mathcal{P}_{GE_2^\bC}(x,y)dxdy=\frac{1}{\pi \left(1+|\la |^2\right)^2}dxdy\, 
$$
can  also be obtained by  explicit calculations with the discriminantal equation similar to those in  Sections 4 - 6. 
\end{remark}

\medskip
Let us now present a number of generalisations of Theorem~\ref{th:GE}. 

\begin{proposition}\label{prop:diag} Conclusion of Theorem~\ref{th:GE} holds, if  $A$ and $B$ are independently chosen from the scaled complex Gaussian ensemble $GE_{\si^2,n}^\bC,$ i.e., the ensemble 
whose off-diagonal entries are i.i.d. standard normal complex variables and whose diagonal entries are i.i.d.  normal complex variables with an arbitrary fixed positive  variance $\si^2$.
\end{proposition}
(In the above notation,  $GE_n^\bC=GE_{1,n}^\bC.$) 

\medskip

The next observation together with Theorem~\ref{th:GE} and Proposition~\ref{prop:diag}  �allows us to substantially extend the class of complex Gaussian ensembles whose distribution of  level crossings is  given by \eqref{densGE}, i.e., it is uniform on $\bC P^1$.

\medskip
Take any complex linear subspace $W_n\subset Mat_n^\bC$ such that the product space $W_n\times W_n\subset Mat_n^\bC\times Mat_n^\bC$ is preserved by the action \eqref{eq:rest}. Given $\si>0$, denote by $W_{\si^2,n}$ the space $W_n$ with the measure  induced from the scaled complex Gaussian ensemble $GE_{\si^2,n}^\bC$. 

\medskip
\begin{proposition}\label{prop:general}
In the above notation,  level crossings  of \eqref{pencil} with the random matrices $A$ and $B$ independently chosen from $W_{\si^2,n}$ are uniformly distributed on $\bC P^1,$ i.e., their probability measure is given by the right-hand side of \eqref{densGE}.   
\end{proposition}

To give an  example of such $W$, recall that $GOE_n^\bC$ is the distribution on the space $Sym_n^\bC$ of complex-valued symmetric matrices, where each entry $e_{i,j}=e_{j,i},\; i < j$ of a $n\times n$-matrix has a normal distribution $N(0,1/2)+i  N(0,1/2)$, and each diagonal entry $e_{i,i}$ is distributed as $\sqrt{2} (N(0,1/2)+iN(0,1/2))$.   Observe that  $GOE_n^\bC$ is obtained by restriction of  $GE_{2,n}^\bC$ to $Sym_n^\bC$.  (Discussions of general spectral properties of complex symmetric matrices can be found  in e.g., \cite {RaGaPrPu}.)  


\begin{corollary}\label{cor:diagsym} Conclusion of  Proposition~\ref{prop:general} holds if $A$ and $B$ are independently chosen from the ensemble $GOE_{n}^\bC,$ and, more generally, from the scaled ensemble  $GOE_{\si^2,n}^\bC$  
whose off-diagonal entries are the i.i.d. standard symmetric normal complex variables and whose diagonal entries are the i.i.d.  normal complex variables with an arbitrary fixed positive variance $\si^2$.
\end{corollary}

\begin{remark}  Further interesting examples of linear subspaces $W$ covered by Proposition~\ref{prop:general}  include Toeplitz matrices, band matrices, band Toeplitz matices, diagonal matrices, etc.  
\end{remark}

\begin{proof}[Proof of Proposition~\ref{prop:diag}]
In the set-up of this Proposition, the density of the probability to obtain a given matrix $A\in Mat_n^\bC$ with respect to the Lebesgue measure is given by the formula
$$\widetilde\ga(A)=Ke^{-\sum_{i\neq j}|A_{ij}|^2-W\sum_{i=1}^n|A_{ii}|^2}=Ke^{-Tr(AA^\star)-W\sum_{i=1}^n|A_{ii}|^2},$$
where $K$ is a  normalisation constant and $W$ is a real number. (To present a probability density in the above formula, the quadratic form $Tr(AA^\star)+W\sum_{i=1}^n|A_{ii}|^2$ has to be positive-definite which implies that $W$ can not be a large negative number.) Therefore
\begin{equation}\label{eq:mod}
\widetilde\ga^{(2)}(A,B)=K^2e^{-Tr(AA^\star +BB^\star)-W\sum_{i=1}^n(|A_{ii}|^2+|B_{ii}|^2)}.
\end{equation}
All we need to show is that the right-hand side of \eqref{eq:mod} is preserved under the action \eqref{eq:action}. In notation of the previous proof, we already know  that $Tr(CC^\star+DD^\star)=Tr(AA^\star+BB^\star)$. It remains to prove that 
$$\sum_{i=1}^n(|A_{ii}|^2+|B_{ii}|^2)=\sum_{i=1}^n(|C_{ii}|^2+|D_{ii}|^2).$$
In fact, $|A_{ii}|^2+|B_{ii}|^2=|C_{ii}|^2+|D_{ii}|^2$ for each $i$ which follows from the relation 
$$|C_{ii}|^2+|D_{ii}|^2=(uA_{ii}+vB_{ii})(\bar u \bar A_{ii}+\bar v \bar B_{ii})+(-\bar vA_{ii}+\bar u B_{ii})(-v \bar A_{ii}+ u \bar B_{ii}) =|A_{ii}|^2+|B_{ii}|^2.$$
\end{proof}

\begin{proof}[Proof of Proposition~\ref{prop:general}] Repeats the above proof of Proposition~\ref{prop:diag}. 
\end{proof}

\begin{proof}[Proof of  Corollary~\ref{cor:diagsym} ]

Both statements follow from the observation that the action \eqref{eq:action} preserves the subspace $Sym_n^\bC\times Sym_n^\bC \subset Mat_n^\bC\times Mat_n^\bC$ and that, additionally, the probability measure of the ensemble $GOE_{\si^2,n}^\bC$ (supported on $Sym_n^\bC\times Sym_n^\bC$)  is induced from that of $GE_{(\si^\prime)^2,n}^\bC$ for appropriate $\si^\prime$.
\end{proof}

\section{$SO_2$-action for $GOE$-, $GUE$- and $GE^\bR$-ensembles}\label{sec3}

This section provides some preliminary material for our study of level crossings of \eqref{pencil} with $A$ and $B$ chosen from the $GOE$-, $GUE$- and $GE^\bR$-ensembles. 
A very essential feature of all these cases is that their level crossings distribution is invariant under the action of the subgroup $SO_2\subset SU_2$ given by the same formula \eqref{eq:rest}, but with real $u$ and $v$ satisfying  $u^2+v^2=1$, see Lemma~\ref{lm:azimuth}.  

\medskip
In the above realization of $\bC P^1$ as the unit sphere  $S^2\subset \bR^3$,  $SO_2$ acts on it by rotation around the $Y$-axis, see Figure~\ref{Fig:projection} and Lemma~\ref{lm:azimuth}�   below. This circumstance implies  that the family of orbits of the $SO_2$-action on the unit sphere $S^2\simeq\bC P^1$ projected to the complex plane of parameter $\la=x+iy$ will coincide with  the family of circles given by 
$$x^2+(y-t)^2=t^2-1,\quad |t|\ge 1.$$

Besides    the above cylindrical  coordinates $(\rho, \phi,Z)$  in $\bR^3$,  let us   introduce  the cylindrical coordinates $(\rho, \psi, Y)$ where $X=\rho \cos \psi,\; Y=Y,\; Z=\rho \sin \psi$. Then $(\psi,Y)$,  $0\le \psi \le 2\pi,\; -1\le Y\le 1$ again parameterises the unit sphere $S^2\simeq \bC P^1$. Lemma~\ref{lm:azimuth} implies that  in the cylindrical coordinates $(\psi,Y)$, the distributions of  level crossings of the above ensembles on $\bC P^1$  are of the form:  
$$dens(\psi,Y)d\psi dY=\rho(Y)d\psi dY,$$
for some  univariate function $\rho$, 
i.e.,  its density  depends only on $Y$ and is independent of the angle variable $\psi$. (In general, $\rho(Y)dY$ can be a $1$-dimensional measure which  does not have a  smooth density function. This happens, for example,  in the case of $GE_2^\bR$, when $\rho(Y)dY$ has a point mass at the origin.)  In the original coordinates $(x,y)$, where $\la=x+iy$, the distribution of  level crossings for the above cases  will be of the form
\begin{equation}\label{eq:form}
dens(x,y)dxdy=\rho\left(\frac{2y}{x^2+y^2+1}\right) \frac{4dx dy}{(x^2+y^2+1)^2},
\end{equation}
with the same $\rho$ as above, see Proposition~\ref{prop:vital}.

\medskip
Therefore the problem of finding the distribution of level crossings for Gaussian orthogonal, Gaussian unitary, and real Gaussian ensembles  becomes in a sense one-dimensional which is a big advantage.  In the cases under consideration, $\rho$ has an additional property of being an even function.

\begin{center} 
						 \begin{figure}[h]
				\includegraphics[scale=0.5]{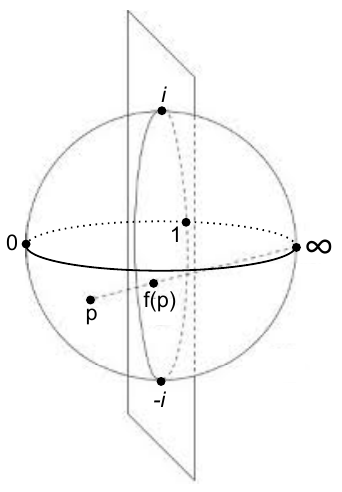}
				\caption{The $SO_2$-action on $\bC P^1$ projectivising the complex plane of parameter $\la$ for the  Gaussian orthogonal, Gaussian unitary, and real Gaussian ensembles.}\label{Fig:projection}
			\end{figure}
\end{center}

We start with the following statement generalizing Lemma~\ref{lm:preserv}. 

\begin{lemma}\label{lm:azimuth} The action of $\mathcal U=\begin{pmatrix} u & -v\\ v & u \end{pmatrix}\in SO_2\subset SU_2$ on pairs of matrices $(A,B)$ given by
$$
(A,B) \ast \begin{pmatrix} u &-v\\ v & u \end{pmatrix} =(uA +v B, -vA+u B),
$$
where $u$ and $v$ are real numbers satisfying the condition $u^2+v^2=1$, 
preserves the following measures on the following matrix (sub)spaces:

\medskip

\noindent 
a) the product $\ga_{GOE}^{(2)}$ of two $GOE_{n}$-measures $\ga_{GOE}$ on the space $Sym_n^\bR\times Sym_n^\bR$;

\noindent
b)  the product $\ga_{GUE}^{(2)}$ of two $GUE_{n}$-measures $\ga_{GUE}$ on the space 
$\HH_n\times \HH_n$;

\noindent
c) the product $\ga_{GE}^{(2)}$ of two $GE_{n}^\bR$-measures $\ga_{GE}$ on the space $Mat_n^\bR\times Mat_n^\bR$. 
\end{lemma}  

\begin{proof} Similarly to   Lemma~\ref{lm:preserv}, $SO_2$ acts on $\widehat {\mathfrak D}_n\subset Sym_n^\bR \times Sym_n^\bR \times \bCP^1$  (resp. on  $\widehat {\mathfrak D}_n\subset \HH_n \times \HH_n \times \bCP^1$  and  on $\widehat {\mathfrak D}_n\subset Mat_n^\bR\times Mat_n^\bR\times \bCP^1$), where $\widehat {\mathfrak D}_n$ is the spectral determinant, i.e.,  the set   of all triples $\left(A,B, (a : b)\right)$ such that $(a : b)$ is a level crossing point of the pair $(A,B)$. (By a slight abuse of notation, in all cases we use the same letter for the spectral determinant.) Here $SO_2$ acts on 
 $\bC P^1$ as 
$$(a:b)\ast \begin{pmatrix}u &-v\\v& u\end{pmatrix} =(ua +vb: -va +ub).$$

Notice that  $(ua+  vb :  -va+ub)$ is a
level crossing point of the pair $(u A + v B, -v  A + a B)$. Indeed,
$$(ua +vb)(u A + v B) + (-va + ub)(-v A + u B)=$$
$$= u^2aA +  v^2bB + a uvB +  bvuA +  v^2aA  + u^2bB - auvB  -bvuA  = aA + bB.$$
Hence $SO_2$ acts on $\widehat {\mathfrak D}_n$, and this action commutes with the projections $\pi_n: \widehat {\mathfrak D}_n \to Sym_n^\bR \times Sym_n^\bR$  (resp.    $\pi_n: \widehat {\mathfrak D}_n \to \HH_n \times \HH_n$  , and   $\pi_n: \widehat {\mathfrak D}_n \to Mat_n^\bR \times Mat_n^\bR$),  
as well as with $\kappa_n : \widehat {\mathfrak D}_n \to  \bC P^1$. To check that  the action of $SO_2$ on  $Sym_n^\bR \times Sym_n^\bR$,  $\HH_n \times \HH_n$, and $Mat_n^\bR \times Mat_n^\bR$, preserves the densities   $\ga_{GOE}^{(2)}$,  $\ga_{GUE}^{(2)}$, and $\ga^{(2)}_{GE}$, respectively,  recall that these densities are given by  $C_{GOE_n}e^{\frac{-n}{4} tr(A^2+B^2)}$,   $C_{GUE_n}e^{\frac{-n}{2} tr(AA^*+BB^*)}$, and $C_{GE_n}e^{\frac{-n}{2} tr(AA^{T}+BB^{T})}$,    
respectively. Here $C_{GOE_n}, C_{GUE_n}, C_{GE_n}$ are the corresponding normalising constants. 

\medskip
Therefore, in e.g., the orthogonal case,  the density of the pair $(A, B)$ is determined by $tr(A^2 + B^2)$. At the same time 
$$tr((u A + v B)^2 + ( -v  A + u B)^2)=$$
$$=tr(u^2A^2 +uv AB+uv BA+ v^2B^2 + v^2A^2-uv  AB -uv BA+u^2B) = tr(A^2 + B^2).
$$
Similar calculations work in the other two cases. 

\medskip  
The density $\mu$ of  level crossing points in $\bCP^1$ is given by $\kappa_n\left(\pi_n^{-1}(\ga_{GOE}^{(2)})\right)$ on  $Sym_n^\bR\times Sym_n^\bR$,   $\kappa_n\left(\pi_n^{-1}(\ga_{GUE}^{(2)})\right)$ on $\HH_n\times \HH_n$, and $\kappa_n\left(\pi_n^{-1}(\ga_{GE}^{(2)})\right)$ on $Mat_n^\bR\times Mat_n^\bR$  resp. That is, the measure $ \mu$ of a measurable set $E \subset \bCP^1$ is given by $\gamma^{(2)}(\pi_n(\kappa_n^{-1}(E)))$.   
Notice that 
$$ \mu (g \cdot E) = \ga^{(2)}(\pi_n(\kappa_n^{-1}(g\cdot E))) = \ga^{(2)}(\pi_n(g\cdot \kappa_n^{-1}(E))) = \ga^{(2)}(g\cdot \pi_n(\kappa_n^{-1}(E))) =$$ 
$$=\ga^{(2)}(\pi_n(\kappa_n^{-1}(E))) = \mu(E).$$
  So we can conclude that for the above three ensembles, the density of  level crossing points on $\bCP^1$ is invariant under the above action by $SO_2$. \end{proof}

\begin{proposition}\label{prop:vital} In the standard coordinates $(X,Y,Z)$ in $\bR^3$ introduced in Remark~\ref{rmk2}, the group 
$SO_2$ acts on $\bC P^1\subset \bR^3$ by rotation with respect to the $Y$-axis. This fact implies that in the above three cases, the distribution of  level crossings in the cylindrical coordinates $(\psi;Y)$ is independent of $\psi$.
\end{proposition}
\begin{proof} We will show that for $\mathfrak U=\begin{pmatrix} \cos \Theta & -\sin \Theta\\  \sin \Theta & \cos \Theta\end {pmatrix}$, its action on a triple $(A, B, (\psi, Y))$ will be given by 
 $$(A, B, (\psi, Y))\ast \mathfrak U = (u A +v B,  -v A + u B, (\psi+  2\theta , Y))$$
implying that  the action of $SO_2$ on $\bC P^1$ realized as the unit sphere in $\bR^3$ is by rotation of the sphere about the $Y$-axis. We only need to concentrate on the action of $\mathfrak U$ on the last coordinate. 
In the homogeneous coordinates $(a: b)$ of  $\bC P^1$, this action, by definition, is given by 
$$(a:b)\ast  \begin{pmatrix} \cos \Theta & -\sin \Theta\\  \sin \Theta & \cos \Theta\end {pmatrix}=(a\cos \Theta+b\sin \Theta: -a\sin \Theta +b \cos \Theta).$$
Setting $\la=\frac{a}{b}$ and $\la=x+iy$, we get that 
$$\la_\Theta:= \la \ast \mathfrak U=\frac{\la \cos \Theta +\sin \Theta}{\cos \Theta-\la \sin \Theta}.$$
In terms of the pair $(x,y)$, the same action is expressed as 
$$(x,u)  \ast \mathfrak U : =(x_\Theta, y_\Theta)=$$
$$\left (\frac{(\sin \Theta + x\cos\Theta)(\cos \Theta -x \sin \Theta)-y^2\sin \Theta \cos \Theta}{(\cos \Theta -x \sin \theta)^2+(y\sin\Theta)^2}, \frac{(\sin \Theta + x\cos\Theta)y\sin \Theta+(\cos \Theta -x \sin \Theta)y\cos \Theta}{(\cos \Theta -x \sin \theta)^2+(y\sin\Theta)^2} \right).$$

\medskip
The relations between the coordinates $(x,y)$ in the $\la$-plane and the coordinates $(X,Y,Z)$ restricted to the sphere are as follows
\begin{equation}\label{eq:rel}
X=x(1-Z)=\frac{2x}{x^2+y^2+1},\quad Y=y(1-Z)=\frac{2y}{x^2+y^2+1},\quad Z=\frac{x^2+y^2-1}{x^2+y^2+1}.
\end{equation}

We have the relation $$(\psi, Y)=\left(\arctan \frac{Z}{X}, Y\right),$$ where $(X,Y,Z)$ are restricted to the sphere.   

\medskip
  We need to express the above $SO_2$-action in the cylindrical coordinates $(\psi,Y)$ on $S^2\simeq \bC P^1$. 
  First we check that the coordinate $Y$  is preserved. In other words,  for any real pair $(x,y)$,  one forms the triple $(X,Y,Z)$ using  \eqref{eq:rel}. Then for the above pair $(x_\Theta,y_\Theta)$, one forms the triple 
  $(X_\Theta,Y_\Theta,Z_\Theta)$ using  \eqref{eq:rel}. What we need to check is that, for any $\Theta$,  one has that $Y=Y_\Theta$. Indeed, $Y_\Theta$ is given by 
  $$Y_\Theta=2\frac{((x C+ S)y S+(C-xS)yC)((C-xS)^2+(yS)^2}{Exp}, $$
  where $C:=\cos \Theta$, $S:=\sin\Theta$, and 
  $$Exp=((xC+S)^2(C-xS)^2-2y^2SC(xC+S)(C-xS)+y^2S^2C^2$$
  $$+(xC+S)^2y^2S^2+2(xC+S)(C-xS)y^2SC+(C-xS)^2y^2C^2$$
  $$+(C-xS)^4+2(C-xS)^2y^2S^2+y^4S^4.$$
  Simplifying the above formula for $Y_\Theta$, we get 
  $$Y_\Theta=\frac{2y}{x^2+y^2+1}=Y.$$
  
  Now we want to find the relation between the angle $\psi_\Theta$ and the pair $(\psi, \Theta)$. Observe that 
  $$\tan \psi_\Theta=\frac{Z_\Theta}{X_\Theta}=\frac{x_\Theta^2+y_\Theta^2-1}{2x_\Theta},$$
  which using the above expressions for $(x_\Theta,y_\Theta)$ gives
  $$\tan \psi_\Theta=\frac{((xC+S)(C-xS)-y^2SC)^2+((xC+S)yS+(C-xS)yC)^2-((C-xS)^2+y^2S^2)^2}{2((C-xS)^2+y^2S^2)((S+xC)(C-xS)-y^2SC)}. 
 $$
Simplifyng the latter expression, we  obtain   
 $$\tan \psi_\Theta=\frac{(x^2+y^2-1)\cos 2\Theta+2x\sin 2\Theta}{2x\cos 2\Theta -(x^2+y^2-1)\sin 2\Theta}=\frac{Z\cos 2\Theta+X\sin 2\Theta}{X\cos 2\Theta -Z\sin 2\Theta}.$$
 Diving the numerator and denominator of the latter expression by $X\cos 2\Theta$, we get 
 $$\tan \psi_\Theta=\frac{\frac{Z}{X}+\tan 2\Theta}{1-\frac{Z}{X}\tan 2\Theta}=\frac{\tan \psi +\tan 2\Theta}{1-\tan \psi \tan 2\Theta}=\tan (\psi+2\Theta), $$
 which implies that $\psi_\Theta=\psi+2\Theta$.
\end{proof}

\begin{lemma}\label{lm:radial} If a smooth distribution which is invariant under the above $SO_2$-action is also radial in the $\la$-plane, then it is constant with respect to the spherical metric on $\bC P^1$.

\end{lemma}

\begin{proof}
Indeed, by formula \eqref{eq:form}, such a distribution  in the $\la$-plane should be of the form 
$$dens(x,y)dxdy=\rho\left(\frac{2y}{x^2+y^2+1}\right) \frac{4dx dy}{(x^2+y^2+1)^2}.$$
On the other hand,  in the polar coordinates $(r,\theta)$ in the $\la$-plane, the same distribution has the form 
$$den(r,\theta)drd\theta=R(r)drd\theta, $$
implying that 
$$\rho\left(\frac{2y}{x^2+y^2+1}\right) \frac{4}{(x^2+y^2+1)^2}=\frac{R(r)}{r} \Leftrightarrow \rho\left(\frac{2y}{r^2+1}\right)=F(r). $$ 
The l.h.s is a function constant on the family of circles 
$$x^2+(y-t)^2=t^2-1,\quad |t|\ge 1$$
while the r.h.s is constant on the family of circles 
$$x^2+y^2=K$$
which can only happen when both sides are constant. Since $\rho \left(\frac{2y}{r^2+1}\right)=K$, the statement follows.  
\end{proof}

\section{Gaussian orthogonal ensembles and Gaussian unitary ensembles}\label{sec4}

Here we prove Theorems~\ref{th:GOE} and ~\ref{th:GUE} stated in the Introduction. The main argument is  similar to our other proofs dealing with the case $n=2$, comp. \cite{ShZa1} and the next section; it has an advantage that one obtains more detailed information. 

Notice that the ensemble $GOE_n$ is invariant under the conjugation by orthogonal matrices implying that for any pair of $GOE_n$-matrices $(A,B)$, we can conjugate $A+\la B$ by an orthogonal matrix to make $A$  diagonal. 

\begin{proof}[Proof of Theorem~\ref{th:GOE}]
By the above, we assume without loss of generality that $A$ is a diagonal matrix, i.e.,  $A=\left(\begin{array}{cc} \al_1 &0 \\
																																																					0 & \al_2 \end{array}\right)$, 
where  $\al_1$ and $\al_2$ are the eigenvalues of $A$ satisfying the condition $\al_1\le \al_2$.
Moreover, we can shift our matrix family so that $A = \left(\begin{array}{cc} 0 &0 \\
																																																					0 & \Delta \end{array}\right)$, where $\Delta = \al_2-\al_1\ge 0$.
																																																					 
We know that  level crossing points of the linear family $A+\la B$ are exactly the zeroes of the discriminant  $Dsc(\la)$ of the characteristic polynomial $\chi(\lambda, t)$ with respect to  the variable $t$, where 
\begin{equation}\label{eq:chi}
\chi(\lambda, t)= \text{det}(A+\la B+t I) = t^2+t(\la\; \mathrm{Tr}\, (B) +\Delta)+\la^2 \det (B) +\la b_{11}\Delta.
\end{equation}

The latter discriminant equals
\begin{equation}\label{eq:dscr}
Dsc(\la)=\la^2((b_{22}- b_{11})^2+4|b_{12}|^2)+ 2\la\Delta(b_{22}-b_{11})+\Delta^2.
\end{equation} 
Therefore, since all coefficients of the latter equation are real and the discriminant of $Dsc(\la)$ considered as a quadratic equation in $\la$ is given by 
$$D=-4\Delta^2|b_{12}|^2\le 0,$$
  level crossing points of a generic pair $(A,B)$ form a complex conjugate pair $(\la,\bar{\la})$, where 											
\begin{equation}\label{eq:tau} 
\la = \Delta\ds\frac{b_{11}-b_{22} + 2i |b_{12}|}{(b_{22}-b_{11})^2 + 4|b_{12}|^2}\quad\text{and} \quad \bar{\la} = \Delta\ds\frac{b_{11}-b_{22} - 2i |b_{12}|}{(b_{22}-b_{11})^2 + 4|b_{12}|^2}.
\end{equation} 

\smallskip
In order to find the distribution of $\la$, we will first find its conditional distribution  assuming that $\Delta$ is constant.
Set $\Sigma := \frac{b_{11}-b_{22}}{\Delta}$ and $\Theta := \frac{2|b_{12}|}{\Delta}$ giving $\la = \frac{1}{\Sigma-i\Theta}$.

\medskip
Since $b_{11}, b_{22} \sim N(0,2)$ and are independent, we get that $\Sigma\sim N(0,\frac{4}{\Delta^2})$. Further,  $\Theta\sim \frac{2}{\Delta}| N(0,1)|$, which can be expressed using $\chi_1$-distribution, see e.g. \cite{chi}.  
Therefore, the conditional PDFs of $\Sigma$ and $\Theta$ are given by 
$$\PP_\Sigma^\Delta (u)=\frac{\Delta}{2\sqrt{2\pi}}\cdot e^{-\frac{u^2\Delta^2}{8}}$$ and  
$$\PP_\Theta^\Delta (v)=\begin{cases} \frac{\Delta}{\sqrt{2\pi}}\cdot e^{-\frac{v^2\Delta^2}{8}}, \text{  for  } v\ge 0;\\
                                                                    0, \quad\quad\quad \quad \text{otherwise.} \end{cases}$$ 
Since $\Sigma$ depends on $b_{11}$ and $b_{22}$, while $\Theta$ depends of $b_{12}$,  we get that $\Sigma$ and $\Theta$ are independent random variables. 
Therefore,  their joint distribution is given by
$$\PP^\Delta_{(\Sigma, \Delta)} (u,v)=\PP_\Sigma^\Delta (u)\cdot \PP_\Theta^\Delta (v)=\begin{cases} \frac{\Delta^2}{4\pi}e^{-\frac{\Delta^2(u^2+v^2)}{8}}, \text{ for }\; v \ge 0;\\
0, \quad\quad\quad\quad\quad \text{otherwise}. \end{cases}$$ 

\medskip  
Introduce $\frak X:=\frac{\Sigma}{\Sigma^2+\Theta^2}$ and $\frak Y:=\frac{\Theta}{\Sigma^2+\Theta^2}$ implying that $\la=\frac{1}{\Sigma-i\Theta}=\frak X+i\frak Y$. Since  the Jacobian of the variable change is given by $$\left|\frac{\partial(x,y)}{\partial(u,v)}\right| = \frac{1}{(u^2+v^2)^2} = (x^2+y^2)^2,$$  the joint distribution  of $\frak X$ and $\frak Y$ coincides with 
 $$\PP_{(\frak X,\frak Y)}^\Delta (x,y)=\begin{cases} \frac{\Delta^2}{4\pi(x^2+y^2)^2}\cdot e^{-\frac{\Delta^2}{8(x^2+y^2)}},  &\text{ for }\; y \ge 0;\\ 0, \quad\quad\quad\quad\quad &\text{otherwise}. \end{cases}$$

Therefore, the conditional distribution of $\la$ with $\Delta$ fixed equals 
$$\PP^\Delta (\la)=\frac{\Delta^2}{4\pi\abs{\la}^4}\cdot e^{-\frac{\Delta^2}{8\abs{\la}^2}},$$

\medskip
The distribution of pairs of eigenvalues $(\al_1,\al_2)$ with $\al_1\leq \al_2$ of a $GOE_2$-matrix is given by  $$\PP(\al_1,\al_2)=\frac{(\al_2-\al_1)}{4\sqrt{2\pi}}\cdot e^{-\frac{\al_1^2+\al_2^2}{4}},$$
where $-\infty<\al_1\le \al_2<\infty$. 

\medskip
\noindent Thus, the distribution of $\la$ with $\text{Im}\;\la >0$ is given by 
\begin{align*}
 \PP_{>0}(\la)&= \iint_{-\infty<\al_1\le \al_2<\infty}
  \frac{(\al_2-\al_1)}{4\sqrt{2\pi}}\cdot e^{-\frac{\al_1^2+\al_2^2}{4}}\cdot \frac{(\al_2-\al_1)^2}{4\pi\abs{\la}^4}\cdot e^{-\frac{(\al_2-\al_1)^2}{8\abs{\la}^2}}\ d\al_2 \  d   \al_1\\
  &= \iint_{-\infty<\al_1\le \al_2<\infty}
  \frac{(\al_2-\al_1)^3}{16\sqrt{2}\cdot \pi^{3/2}} \cdot  \frac{1}{\abs{\la}^4} \cdot e^{-\frac{\al_1^2+\al_2^2}{4}}\cdot e^{-\frac{(\al_2-\al_1)^2}{8\abs{\la}^2}}\ d\al_2 \  d   \al_1
\\
 &=\frac{2}{\pi(1+\abs{\la}^2)^2}.
\end{align*}

To get the actual PDF of $\la,$ we must  divide the previous answer by $2$, getting $$\PP_{GOE_2}(\la)=\frac{1}{\pi(1+\abs{\la}^2)^2}.$$ 
\end{proof}

Now we consider the $2\times 2$-Gaussian unitary ensemble. 

\begin{proof}[Proof of Theorem~\ref{th:GUE}] 

Using the same methods as for $GOE_2$, we calculated the distribution of level crossings for $GUE_2$-case. 
As in the previous case,  level crossing point $\la$ with nonnegative imaginary part is given by   $$\la =\Delta \frac{b_{11}-b_{22}+2i b_{12}}{(b_{22}-b_{11})^2+4\mid b_{12}\mid^2}=\frac{1}{\Sigma-i\Theta},$$
where $\Sigma:=\frac{b_{11}-b_{22}}{\Delta}$ and
$\Theta:=\frac{2\abs{b_{12}}}{\Delta}$.

\medskip
Since $b_{11},b_{22}\sim N(0,1)$ and are independent, we obtain $b_{22}-b_{11}\sim N(0,2)$, and hence, $\Sigma\sim N\left(0,\frac{2}{\Delta^2}\right)$.
Therefore, the conditional PDF of $\Sigma$ is given by 
$$\PP_\Sigma^{\Delta}(u) = \frac{1}{\frac{\sqrt{2}}{\Delta} \sqrt{2}\sqrt{\pi}}e^{-u^2\Delta^2/4} = \frac{\Delta}{2\sqrt{\pi}}e^{-u^2\Delta^2/4}.$$
Since $\text{Re}(b_{12}),\text{Im}(b_{12})\sim N(0,\frac{1}{2})$, then $\frac{1}{1/\sqrt{2}} |b_{12}| \sim \chi_2$. 
Thus, the conditional PDF of $\Theta$ is given by  
$$\PP_\Theta^{\Delta}(v) =\begin{cases} \frac{\Delta^2y}{2}e^{-v^2\Delta^2 /4}, \text{  for  } v\ge 0;\\
                                                                  0, \quad\quad\quad\quad \text{         otherwise.} \end{cases}$$
The joint distribution of $\Sigma$ and $\Theta$ gives us the conditional distribution of $\frac{1}{\la}$. 
Since $b_{12}$ is independent of $b_{11}$ and $b_{22}$, then $\Sigma$ and $\Theta$ are also independent random variables which implies that the conditional PDF of $\frac{1}{\la}$ is the product of the PDFs of $\Sigma$ and $\Theta$, i.e., 
$$\PP_{1/\la}^{\Delta}(u,v)=\PP_\Sigma^{\Delta}(u)\cdot \PP_\Theta^{\Delta}(v) = \frac{\Delta^3}{4\sqrt{\pi}} ve^{-(u^2+v^2)\Delta^2 /4}.$$
Introducing $\frak X := \frac{\Sigma}{\Sigma^2+\Theta^2}$ and $\frak Y := \frac{\Theta}{\Sigma^2+\Theta^2}$, we get  $\la = \frac{1}{\Sigma-i\Theta} = \frak X + i\frak Y$. 
Since  the Jacobian of the variable change is given by $$\left|\frac{\partial(x,y)}{\partial(u,v)}\right| = \frac{1}{(u^2+v^2)^2} = (x^2+y^2)^2,$$  the joint distribution  of $\frak X$ and $\frak Y$ coincides with  
$$\PP^\Delta_{(\frak X,\frak Y)}(x,y) = \frac{\PP^\Delta_\Sigma(u)\cdot \PP^\Delta_\Theta(v)}{\left|\frac{\partial(x,y)}{\partial(u,v)}\right|} = \frac{y \Delta^3 e^{-\Delta^2/4(x^2+y^2)}}{4\sqrt{\pi}(x^2+y^2)^3}.$$
As $\frak X= \text{Re}(\la)$ and $\frak Y=\text{Im}(\la)$, then for a given value of $\Delta$, the conditional distribution of $\la$  is given by 
$$\PP^{\Delta}(\la)=\frac{\text{Im}(\la) \Delta^3 e^{-\Delta^2/4|\la|^2}}{4\sqrt{\pi}|\la|^6}.$$
\medskip
Finally, in order to find the (unconditional) distribution of $\la$, we recall that the PDF of the joint distribution for pairs of eigenvalues $\al_1\le \al_2$ of a random matrix  belonging to $GUE_2$  is given by 
$$\PP(\al_1,\al_2) = \frac{1}{2\pi}(\al_2-\al_1)^2 e^{-(\al_1^2+\al_2^2)/2}=\frac{\Delta^2 }{2\pi}e^{-(\al_1^2+\al_2^2)/2}.$$
Therefore, since $\Delta=a_2-a_1$, the  distribution for  level crossing point $\la$ with $\text{Im}\;\la \ge 0$ is given by  
\begin{center}\begin{tabular}{rl}
$\PP_{>0}(\la)$ &  $= \ds\iint_{-\infty<a_1\le a_2 <+\infty}  \PP^{\Delta}(\la)\cdot \PP(\al_1, \al_2) \ d\al_2\ d\al_1$ \\
&  $= \ds\iint_{-\infty<a_1\le a_2 <+\infty} \frac{(\al_2-\al_1)^5 }{8\pi^{3/2}} \cdot \frac{\text{Im}(\la)}{|\la|^6} \cdot e^{\frac{-(\al_2-\al_2)^2}{4|\la|^2}-\frac{(\al_1^2+\al_2^2)}{2}}\ d\al_2\ d\al_1$ \\
& $ = \ds\frac{8\text{Im}(\la)}{\pi(1+|\la|^2)^3}.$\end{tabular}\end{center}
Therefore,  the actual distribution for  level crossing point $\la\in \bC$ equals 
$$\PP_{GUE_2}(\la)= \ds\frac{4|\text{Im}(\la)|}{\pi(1+|\la|^2)^3}.$$
\end{proof}

\section{Real Gaussian ensembles}\label{sec55}

In this section, in order to prove Propositions~\ref{prop:GE2Rparts}  and \ref{prop:GE2R} stated in the Introduction, we will use the standard presentation of real $2\times 2$-matrices as linear combinations of Pauli matrices which was extensively applied in \cite{ShZa1}. Namely,  let $A = (a_+, ia_-, a_\Delta) \cdot \vec \sigma $ be a real 2-by-2 matrix with normal variables, generic up to additional multiples of identity. Here $\vec \sigma=(\sigma_1,\sigma_2,\sigma_3)$ is the standard triple of Pauli matrices. Denote the coefficient vector $(a_+, a_-, a_\Delta)$ by $\vec A$ and consider the inner product  on such triples using a Minkowski metric: 
\begin{align}
\vec A \cdot \vec B :=_{\text{def}} \vec A\pmtx{1 & 0 & 0\\ 0 & -1 & 0\\ 0 & 0 & 1} \vec B.
\end{align}
Notice that the discriminant of $A$, i.e., the expression which vanishes if and only $A$ has a multiple eigenvalue, is  given by 
 \begin{align}
D_A = \vec A \cdot \vec A = a_+^2 - a_-^2 + a_\Delta^2. 
 \end{align}
Similarly construct $B= (b_+, ib_-, b_\Delta) \cdot \vec \sigma$ and consider the linear family
\begin{align}
 C = A + \lambda B.
\end{align} 
We get
\begin{align}
D_C = D_A + \lambda^2 D_B + 2\lambda \vec A \cdot \vec B
\end{align}
 with zeroes at
\begin{align}
\lambda = - \frac{\vec A \cdot \vec B}{D_B} \pm \sqrt{\left(\frac{\vec A \cdot \vec B}{D_B}\right)^2 - \frac{D_A}{D_B}}.
\end{align}

Firstly, let us prove Proposition~\ref{prop:GE2Rparts}. 
\begin{proof}
To settle Part (i), observe  that $\frac{\lambda_+ + \lambda_-}{2} = -\frac{\vec A \cdot \vec B}{\vec B \cdot \vec B}$ which amounts to computing a single delta. In this case we make the isometric transformation $b_- \mapsto -b_-$ and let $a$ be the component of $\vec A$ along $\vec B$, which allows us to work in a 
Euclidean space for the purposes of computing $\vec A \cdot \vec B$. We also use spherical coordinates for $\vec B$ given by:
\begin{align}
b_+ &= R_B\sin\phi_B \cos\theta_B;
\\
b_\Delta &= R_B\sin\phi_B \sin\theta_B;
\\
b_- &= -R_B\cos\phi_B;
\\
-\frac{\vec A \cdot \vec B}{D_B} &= -\frac{a R_B}{R_B^2(1-2\cos^2\phi_B)} = \frac{a}{R_B\cos(2\phi_B)}.
\end{align}
Note that here $R_B^2$ is $\chi_3^2$-distributed whereas $r_B^2$ is $\chi_2^2$-distributed.
 Next observe that $c \equiv \cos\phi_B$ is uniformly distributed for any spherically symmetric distribution, which means that if $c = \pm\sqrt{\frac{1-t}{2}}$, then
\begin{align}
\rho_{\cos\phi_B}(c)dc &= \frac{dc}{2};\\
\rho_{1-2\cos^2\phi_B}(t)dt &= \frac{2\left|\frac{d}{dt}\sqrt{\frac{1-t}{2}}\right|dt}{2}dt = \frac{dt}{2\sqrt{2-2t}};\\
R_B \propto \rho_{\sqrt{\chi_3^2}}(R) &= \sqrt{\frac{2}{\pi}}R^2 e^{-R^2/2}.
\end{align}
So the distribution of the average of two level crossings  simply becomes
\begin{align}
\rho_{\frac{\lambda_+ + \lambda_-}{2}}(x) = \iiint \delta\left(x + \frac{a}{R t}\right)\rho(a,R,t) da dR dt.
\end{align} 
Resolving the delta with respect to $a$ gives $\left|\frac{da}{dx}\right| = |Rt|$ which implies that 
\begin{align}
\begin{split}
 \rho_{\frac{\lambda_+ + \lambda_-}{2}}(x) &= \int_{-1}^1 dt \int_0^\infty dR \left|R^3t\right|\frac{e^{-\frac{R^2 + (xtR)^2}{2}}}{2\pi\sqrt{2-2t}}\\
&= \int_{-1}^1 dt \int_0^\infty dR \frac{|t|}{2\pi\sqrt{2-2t}}R^3e^{-(1+x^2t^2)\frac{R^2}{2}}\\
&= \int_{-1}^{1} \frac{\left| t\right|}{\pi\sqrt{2-2 t} \left(x ^2 t^2+1\right)^2} dt.
\end{split}
\end{align}

To settle Part (ii), compute the distribution of $D_B = (b_+^2 +b_\Delta^2) - b_-^2$:
\begin{align}
\begin{split}
 \rho_D(D) &= \iint_0^\infty \delta(D+y-x)\rho_{\chi_2^2}(x)\frac{e^{-y/2}}{\sqrt{8\pi y}} dx dy\\
 &= \frac{e^{-D/2}}{\sqrt 8}\left(1 - \Theta(-D)\text{erf}(\sqrt{-D})\right).
\end{split} \label{eqn:DiscriminantDistribution}
\end{align}It's worth noting that in the positive range this is just $\frac{\rho_{\chi^2_2}(x)}{\sqrt 2}$, so the probability that the discriminant is positive is $\frac{1}{\sqrt{2}}$.

The distribution of the product $\lambda_+\lambda_- = \frac{D_A}{D_B}$ is the $D$-ratio distribution:
\begin{align}
 \rho_{\lambda_+\lambda_-}(x) = \int_{-\infty}^{\infty} |y|\rho_D(y)\rho_D(xy) dy.
\end{align}
We split the latter integral into four parts depending on the signs of $x$ and $y$:
\begin{align}
 \rho_{++} &= \int_0^\infty \frac{ye^{-y(1+x)/2}}{8} = \frac{1}{2(x+1)^2};\\
 \rho_{+-} &= \int_{-\infty}^0 -\frac{y e^{-y(1+x)/2}}{8}\text{erfc}(\sqrt{-y})\text{erfc}(\sqrt{-xy});\\
 \rho_{--} &= \int_{-\infty}^0 -\frac{y e^{-y(1+x)/2}}{8}\text{erfc}(\sqrt{-y}) = \frac{1}{2(x+1)^2}\left(1+\frac{3x-1}{\sqrt{2}(1-x)^{3/2}}\right);\\
 \rho_{-+} &= \int_0^{\infty} \frac{y e^{-y(1+x)/2}}{8}\text{erfc}(\sqrt{-xy}) = \frac{1}{2(x+1)^2}\left(1+\frac{(x-3) \sqrt{-x}}{\sqrt{2} (1-x)^{3/2}}\right).
\end{align}

Observe that only one integral out of four can not  be computed in a closed form, but it can be computed numerically using e.g.,  Mathematica. Combining terms, we get
\begin{align}
\begin{split}
\rho_{\frac{D_A}{D_B}}(x) &= \Theta(x)\left[\frac{1}{2(x+1)^2} -\int_{-\infty}^0 \frac{y e^{-y(1+x)/2}}{8}\text{erfc}(\sqrt{-y})\text{erfc}(\sqrt{-xy})dy\right] \\
 &+ \Theta(-x) \left[ \frac{1}{(x+1)^2}\left(1+\frac{3x-1 + (x-3) \sqrt{-x}}{\sqrt{8} (1-x)^{3/2}}\right)\right]
\end{split}
\end{align} 
which is the required expression. 
\end{proof}

 We now turn to Proposition~\ref{prop:GE2R}. 

\begin{lemma}\label{lm:sqrt2}
If $A$ and $B$ are independently chosen from the  $GE_2^\bR$-ensemble, then the probability of attaining a real pair of level crossing points $\lambda_\pm$ in the family $C=A+\la B$ equals  $\frac{1}{\sqrt 2}$.
\end{lemma}
\begin{proof}
We use a result from \cite{ShZa1} saying that the proportion of real eigenvalues for a fixed $A$ is given by
\begin{align}
\kappa(a_+, a_-, a_\Delta) = 
\begin{cases}
1 &\text{if } D_A < 0 \\
1-\frac{1}{\pi}\arccos \frac{a_-^2}{a_+^2 + a_\Delta^2} &\text{if } D_A \geq 0,
\end{cases}
\end{align}
see formula (5.43) in loc. cit. 
The expectation value over the set of matrices with positive discriminant is given by
\begin{align}
\langle \kappa \rangle_{D_A \geq 0} = \iiint_{D_A \geq 0} \left(1-\frac{1}{\pi}\arccos \frac{a_-^2}{a_+^2 + a_\Delta^2}\right) \rho(a_+, a_-, a_\Delta) da_+ da_- da_\Delta.
\end{align}
Using spherical coordinates relative to the $a_-$-axis we can simplify the integral as:
\begin{align}
 \int_0^\infty 2\pi r^2 \frac{e^{-r^2/2}}{(2\pi)^{3/2}}dr \int_{-\frac{1}{\sqrt 2}}^{\frac{1}{\sqrt 2}} \left(1-\frac{1}{\pi}\arccos \frac{\cos^2\phi}{1-\cos^2\phi}\right) d(\cos\phi)
 = \sqrt 2 - 1.
\end{align}
On the other hand, the contribution of the set of matrices with $D_A < 0$ is just 
\begin{align}
\langle \kappa \rangle_{D_A < 0} = \iiint_{D_A < 0} \rho(a_+, a_-, a_\Delta) da_+ da_- da_\Delta = P(D_A < 0) = 1-\frac{1}{\sqrt 2}
\end{align} where the last step follows from equation (\ref{eqn:DiscriminantDistribution}). Thus the total probability of getting a real crossing value is
\begin{align}
 \langle \kappa \rangle = \sqrt 2 - 1 + 1-\frac{1}{\sqrt 2} = \frac{1}{\sqrt 2}.
\end{align}
\end{proof}%

Let us now prove Proposition~\ref{prop:GE2R}. 

\begin{proof}
 Due to the isotropy of a normally distributed vector, we are free to rotate the coordinate system in the $(b_+,b_\Delta)$-plane such that $\vec B = (r_B,b_-,0)$. This $B$-dependent choice of a basis has no impact on the distribution of $\vec A$ which has the normally distributed entries $(a_1, a_-, a_2)$ in this basis.

To settle Part (i) of the Proposition,  assume that the level crossing points $\lambda_\pm = x \pm iy$ are complex conjugate, in which case we get
\begin{align}
 x &= -\frac{\vec A\cdot \vec B}{\vec B \cdot \vec B} = \frac{a_1 r_B - a_- b_-}{r_B^2 - b_-^2};\\  
 y &= \sqrt{\frac{D_A}{D_B}-\left(\frac{\vec A \cdot \vec B}{D_B}\right)^2} = \sqrt{\frac{a_1^2 -a_-^2 + a_2^2}{r_B^2-b_-^2}-x^2}.
\end{align}
Therefore the density of  the joint distribution with respect to the Lebesgue measure in the plane takes the form
\begin{align}
  \rho(x,y) = \iiint_A\iiint_B \delta\left(x+\frac{a_1r_B-a_-b_-}{r_B^2 - b_-^2}\right) \delta\left(y-\sqrt{\frac{a_1^2 -a_-^2 + a_2^2}{r_B^2-b_-^2}-x^2}\right) \rho(\vec A, \vec B). 
\end{align}
Resolving the first delta with respect to $a_-$, we get
\begin{align}
 a_- &= \frac{a_1r_B + x(r_B^2 - b_-^2)}{b_-};\\
 \left|\frac{da_-}{dx}\right| &= \left|\frac{r_B^2 - b_-^2}{b_-}\right|.
\end{align}
Then resolving the second delta with respect to $a_2^2$, we obtain
\begin{align}
 a_2^2 &= (r_B^2 - b_-^2)(x^2 + y^2) + \left(\frac{a_1r_B + x(r_B^2 - b_-^2)}{b_-}\right)^2 - a_1^2;\\
 \left|\frac{d(a_2^2)}{dy}\right| &= \left|2y(r_B^2 - b_-^2)\right|.
\end{align}
Inserting, we get 
\begin{align}
\begin{split}
 \rho(x,y) = \int_{a_1}\iiint_B \left|\frac{2y}{b_-}(r_B^2 - b_-^2)^2\right| \rho_{a_-}\left(\frac{a_1r_B + x(r_B^2 - b_-^2)}{b_-}\right)\\
 \cdot \rho_{a_2^2}((r_B^2 - b_-^2)(x^2 + y^2) + \left(\frac{a_1r_B + x(r_B^2 - b_-^2)}{b_-}\right)^2 - a_1^2)\\
 \cdot \rho_{a_1}(a_1) \rho_{r_B}(r_B)\rho_{\theta_B}(\theta_B) \rho_{b_-}(b_-).
\end{split}
\end{align}
Expanding the expression and integrating out $\theta_B$ gives us:
\begin{align}
\begin{split}
 \rho(x,y) = \int_{-\infty}^\infty da_1 \int_0^\infty dr_B \int_{-\infty}^\infty db_- \left|\frac{2y}{b_-}(r_B^2 - b_-^2)^2\right| \frac{e^{-\left(\frac{a_1r_B + x(r_B^2 - b_-^2)}{b_-}\right)^2/2}}{\sqrt{2\pi}}\\
 \cdot \frac{e^{-((r_B^2 - b_-^2)(x^2 + y^2) + \left(\frac{a_1r_B + x(r_b^2 - b_-^2)}{b_-}\right)^2 - a_1^2)/2}}{\sqrt{2\pi}\sqrt{(r_B^2 - b_-^2)(x^2 + y^2) + \left(\frac{a_1r_B + x(r_B^2 - b_-^2)}{b_-}\right)^2 - a_1^2}}\\
 \cdot \Theta\left[((r_B^2 - b_-^2)(x^2 + y^2) + \left(\frac{a_1r_B + x(r_B^2 - b_-^2)}{b_-}\right)^2 - a_1^2\right] \\
 \cdot \frac{e^{-a_1^2/2}}{\sqrt{2\pi}} r_Be^{-r_B^2/2}\frac{e^{-b_-^2/2}}{\sqrt{2\pi}}.
\end{split}
 \end{align}
 After some extra simplifications,  we get
 \begin{align}
\begin{split} \rho(x,y) = \int_{-\infty}^\infty da_1 \int_0^\infty dr_B \int_{-\infty}^\infty db_- \left|\frac{yr_B}{2\pi^2 b_-}(r_B^2 - b_-^2)^2\right| e^{-\frac{a_1^2+r_B^2+ b_-^2}{2}}\\
 \cdot \frac{e^{-((r_B^2 - b_-^2)(x^2 + y^2) + \left(\frac{a_1r_B + x(r_b^2 - b_-^2)}{b_-}\right)^2 - a_1^2)/2}}{\sqrt{(r_B^2 - b_-^2)(x^2 + y^2) + \left(\frac{a_1r_B + x(r_B^2 - b_-^2)}{b_-}\right)^2 - a_1^2}}\\
 \cdot \Theta\left[((r_B^2 - b_-^2)(x^2 + y^2) + \left(\frac{a_1r_B + x(r_B^2 - b_-^2)}{b_-}\right)^2 - a_1^2\right]. 
 \end{split}
 \end{align}
Suppressing the superfluous subscripts from the integration variables, we obtain the triple integral from the formulation of  Proposition~\ref{prop:GE2R}.  
 
To settle Part (ii), observe that by formula \eqref{eq:form}, the density of a distribution the level crossings invariant under the $SO_2$-action  on the real axis should be proportional to $\frac{1}{(1+x^2)^2}$. By Lemma~\ref{lm:sqrt2} the total mass of the measure of level crossings  concentrated on the real axis equals $\frac{1}{\sqrt{2}}$. Using this normalization, we arrive at the expression \eqref{eq:realaxis}.
 \end{proof}


\section{Monodromy distribution for $3\times 3$ Gaussian ensembles}\label{sec6} 

In this  section we present  numerical results about the monodromy of random $3\times 3$ linear matrix families \eqref{pencil}.   
 Monodromy statistics was  collected for the cases of $GUE_3$-, $GOE_3$-, and $GE_3^\bC$-ensembles. (One can easily check that  the number of possible monodromy sequences for the matrix sizes exceeding $3$ is  so large that  it is  practically impossible to collect coherent statistical information numerically.) Some of the numerical results below are rather surprising, see Remark~\ref{remfin}. 

\medskip \noindent 
{\it General observations.}�  {\rm 
Observe that, for generic pairs of matrices $A$ and $B$ from $GUE_n$ and $GOE_n$, all level crossings are  simple and arise in complex conjugate pairs;  $\binom {n}{2}$  of them  lying  in the upper half-plane and 
$\binom {n}{2}$  lying symmetrically in the lower half-plane. We can additionally assume that all  level crossings in the upper half-plane have distinct real parts since the coincidence of the real parts happens with probability $0$.  
Denote by  $\la_1,\la_2,\dots,\la_{\binom{n}{2}}$   level crossing points in the upper half-plane ordered by the increase of their real parts. Since generically  level crossing points are
simple, let $\si_1, \si_2,\dots, \si_{\binom{n}{2}}$ be the associated sequence of transpositions  obtained as follows, see Fig.~\ref{Fig5}. Under our assumptions, for every real $\la$, the spectrum of $A+\la B$ is real and simple which means that no monodromy of the spectrum occurs when $\la$ belongs to the real axis $\bR\subset \bC$. 

\medskip
If $\la_i$ is the $i$-th level crossing point in the upper half-plane in the order of increasing  real parts, consider  the path in the $\la$-plane  starting on the real axis at $\tau=Re (\la_i)$, going  straight up  to $\la_i$, making a small loop encircling $\la_i$ counterclockwise,  and returning back to $\tau_i$.   As a result, one gets a transposition $\si_i$ of two real eigenvalues corresponding to $\tau_i=Re(\la_i)$. Doing this for each $\la_i$, $i=1,\dots , \binom{n}{2}$, we obtain a sequence of $\binom{n}{2}$ transpositions $\left(\si_1, \si_2,\ldots, \si_{\binom {n}{2}}\right)$,  $\si_i \in S_n$. 

One can easily check that the obtained sequence $\left(\si_1, \si_2,\ldots, \si_{\binom {n}{2}}\right)$ of  transpositions satisfies the following two conditions: 

\medskip
\noindent
(i) for general $A$ and $B$, they generate the symmetric group $S_n$; 

\noindent
(ii) the product $\si_1 \cdot \si_2\cdot \ldots \cdot \si_{\binom {n}{2}}$ coincides with the inverse permutation  $(n, n - 1, \dots ,1)$.

\begin{center}
 \begin{figure}[h]
				
				\includegraphics[scale=0.4]{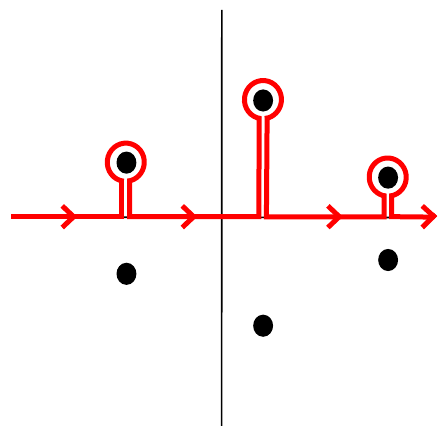}
				\caption{Creating the monodromy sequence}
			\label{Fig5}
				\end{figure}
				\end{center}
				
				Notice that the statistics of the monodromy sequences of transpositions for $GUE_n$ and $GOE_n$ are invariant under  
conjugation by the inverse permutation $(n, n-1, \dots, 1)$ as well as under reversing
the order of the transpositions. These symmetries can be explained as consequences of the symmetries of the ensembles.

Namely, if the matrix $A+\la B$ has eigenvalues $\al_1, \al_2, \dots,  \al_n$, then the matrix $-A-\la B$ has eigenvalues $-\al_1, -\al_2, \dots,  - \al_n$. These matrix pencils share the same level crossing points, and if a loop in $\bCP^1$ permutes the eigenvalues of $A + \la B$, then it applies the same permutation to the eigenvalues of $-A -\la B$. However, when we compute the monodromy associated to a pair of matrices in our ensembles, we order the (real) eigenvalues for real $\la$, and the transpositions associated to each level crossing point are written with respect to this ordering. Since the eigenvalues of $-A - \la B$ will have the  ordering opposite to those of $A + \la B$, the monodromy associated to the pair $(-A, -B)$ will be the monodromy of $(A, B)$, conjugated
by $(n, n -1, \dots , 1)$. Since the pairs $(A, B)$ and $(-A, -B)$ have the same probability density, each of the  admissible sequences of transpositions will appear with the same frequency as its conjugate.

The other symmetry of our data is its invariance under reversing the order of the transpositions. It  can be similarly explained by the equal probability density  for the pairs $(A,B)$ and $(A,-B)$. If  level crossing points of $A + \la B$  are $\la_1, \la_2, \dots , \la_{n(n-1)}$, then  level crossing points of $A - \la B$ are $-\la_1, -\la_2, \dots,  -\la_{n(n-1)}$. Level crossing points come in conjugate pairs, and the same transpositions are associated to these pairs, so if $\la_1, \la_2, \dots, \la_{\binom{n}{2}}$ are  level crossing  points of $A + \la B$ in the upper half-plane, then $-\overline \la_1, -\overline \la_2, \dots, -\overline \la_{\binom{n}{2}}$ are  level crossing points of $A -\la B$ in the upper half-plane. Since we order them according to  the increase of their real parts, which have been inverted, it now remains to show that the transposition associated to $(A, B, \la_i)$ is the same as that associated to $(A, -B, -\overline \la_i)$. Since the transposition associated to  level crossing point is the same as that associated to its conjugate, we can instead consider $(A, -B, -\la_i)$.
Observe that the transposition associated to $(A,B,\la_i)$ is determined by the eigenvalues of
$$A + (Re(\la_i) + \eps Im(\la_i))B$$
for $0 \le \eps \le 1$, and in the same way the transposition associated to $(A, -B, -\la_i)$ is determined
by
$$A + (Re(-\la_i) + \eps Im(-\la_i))(-B) = A + (Re(\la_i) + \eps Im(\la_i))B.$$
These coincide, and we conclude that the monodromy sequence associated to $(A,-B)$ is the reverse of that associated to $(A, B)$.
}

\medskip \noindent
{\it Statistical results for $GUE_3$- and $GOE_3$-ensembles.}
\medskip 

For $n = 3$, it is easy to check that there are only $8$ triples  of transpositions in $S_3$ satisfying conditions (i) and (ii). These triples are: 
$(12)(12)(13);$  $(12)(13)(23);$ $ (12)(23)(12);$ $ (13)(12)(12);$  $(13)(23)(23); (23)(12)(23);(23)(13)(12); (23)(23)(13).$  (For comparison, for $n=4$, there are already 3840 $6$-tuples of transpositions in $S_4$ satisfying (i) and (ii).)

\begin{center} 
 \begin{figure}[h]
				
				\includegraphics[scale=0.2]{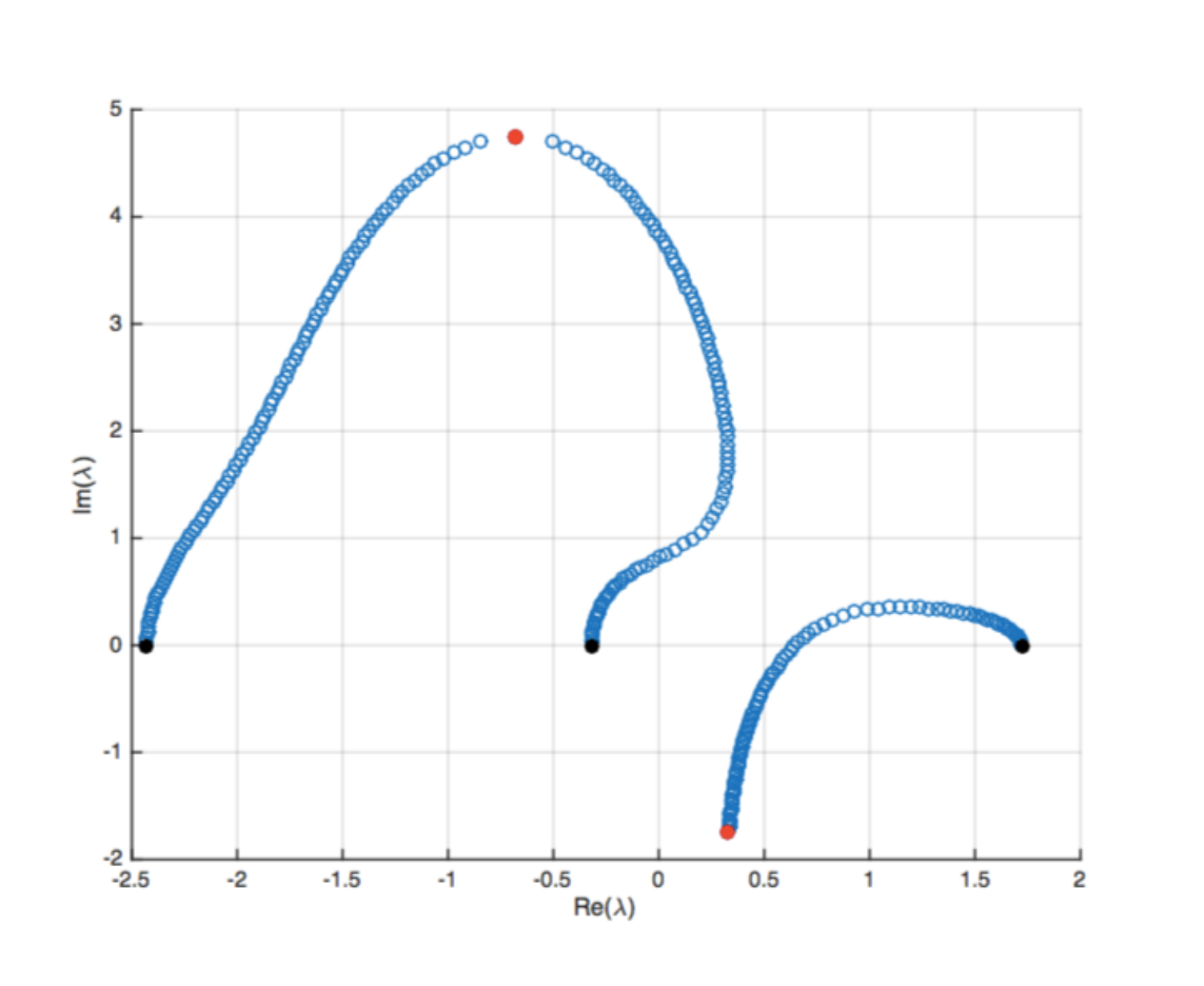}
				\caption{The first and second eigenvalues (of totally 3) collide as the parameter approaches  level crossing point, giving the transposition $(12)$.}
\label{Fig51}
\end{figure}
\end{center}

\begin{center}
 \begin{figure}[h]
				
				\includegraphics[scale=0.32]{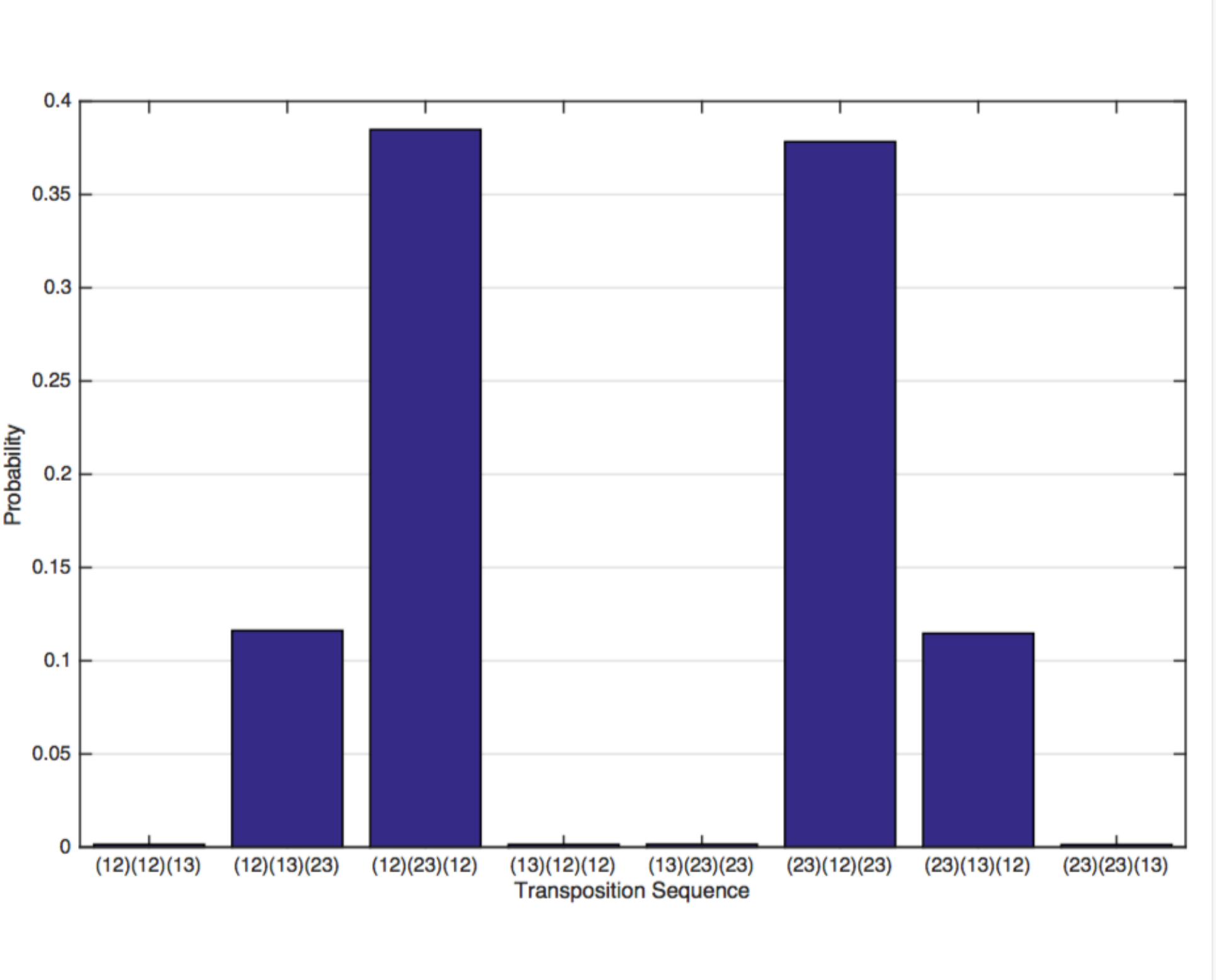}  \includegraphics[scale=0.38]{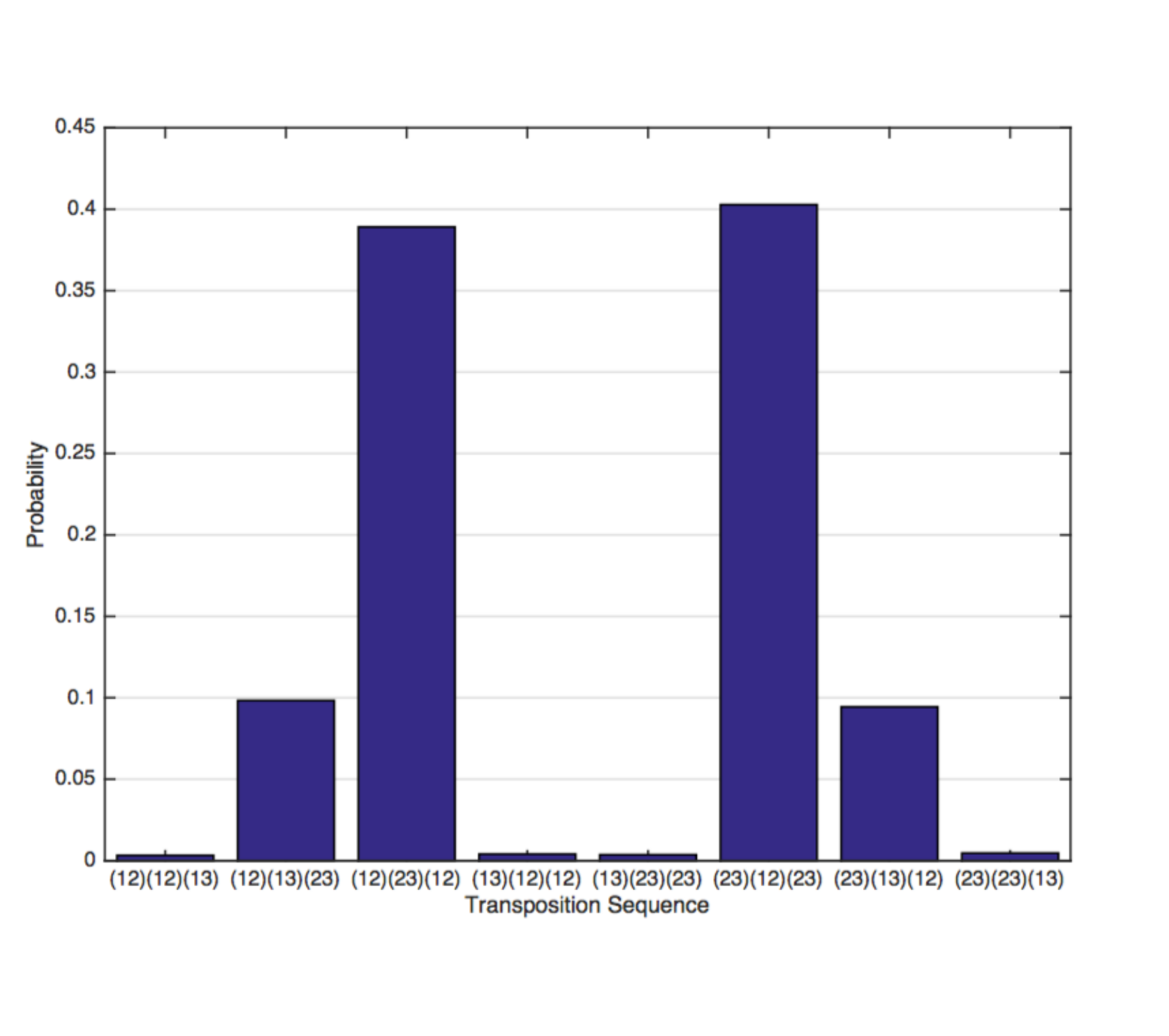}
				\caption{The probabilities of  the monodromy triples of transpositions for $GUE_3$- and $GOE_3$-matrices.}
		 \label{Fig6}
			\end{figure}

	\end{center}	

Numerical experiments were carried out in MATLAB. Namely, the MATLAB-code  computed the transposition associated to  level crossing point $\la$ of a pair of matrices $(A,B)$. More exactly,  the program calculated the eigenvalues of 
$A + (Re(\la)  + \eps Im(\la))B$ as $\eps$ runs from $0$ to $1$ in steps of $0.01$. A typical plot of the eigenvalues during this process is shown in Fig.~\ref{Fig51}. At $\eps = 0$ all of the eigenvalues are real, so we can number them in the increasing order. For each new $\eps$, the new eigenvalues are assigned the same numbers as the closest eigenvalues obtained for the previous value of $\eps$. Then, when two eigenvalues collide at $\eps = 1$, the numbers assigned to these two colliding eigenvalues give the transposition corresponding to  level crossing point $\la$. By following this procedure shown in Figure~\ref{Fig51} for each of  level crossing points in the upper half-plane in order of increasing real part,  one obtains triples of transpositions associated to $(A,B)$. This triple of transpositions complete determines the monodromy of the linear family \eqref{pencil}. Because errors can occur if the real parts of different level crossing points are very close, we discarded such pairs of matrices  when gathering monodromy statistics.  This procedure was carried out  in case of $GUE_3$- and $GOE_3$-ensembles.  The resulting statistics for $GUE_3$ (top) and $GOE_3$ (bottom) are shown in Figure~\ref{Fig6}.

\medskip\noindent						
{\it Statistical results for $GE_3^\bC$-ensemble.}
\medskip

In this case, in order to calculate the monodromy seqeunce for a general matrix family \eqref{pencil}, we must first choose a base point for the system of closed paths in the $\la$-plane which is (generically) not a level crossing point. We choose  $\la = 0$, since typically the origin  is not a level crossing point for a general pair of matrices, and the preimages of $0$ are precisely the eigenvalues of $A$. Using $\la=0$ as a base point, we need to order our level crossing points with respect to the origin  and to choose a system of paths  such that  

\noindent
(i) each path begins and ends at $0$;

\noindent
(ii)  each path goes around exactly one level crossing point;

\noindent
(iii)  each path does not intersect any  other path except at the origin. 

\medskip
As already mentioned, these level crossing points are all generically simple; so as $\la$  traverses a path around one level crossing point and returns to the origin, exactly two of the eigenvalues of $A$ will interchange.  Thus we obtain a transposition in the symmetric group $S_n$. To do this we have to order the preimages of our starting point (i.e., the eigenvalues of $A$) and keep track of how these preimages change as we follow each path. This procedure gives us an $n(n-1)$-tuple of transpositions in $S_n$. Since the concatenation of all paths encompasses all of our level crossing points, the product of all transpositions in the chosen order equals to the identity permutation. 
When $A$ and  $B$ are independently chosen from  $GE_n^\bC$, the arguments of our level crossing points are uniformly distributed, so we may order our level crossing points by the argument.  However the choice of which level crossing point is first and whether the level crossing points are ordered clockwise or counterclockwise is arbitrary. The paths we choose will start and end at $0$ and go around these level crossing points in a natural way. An example of how we choose such paths is shown in Fig.~\ref{Fig7}.

\begin{center}
	\begin{figure}[h]
				
				\includegraphics[scale=0.4]{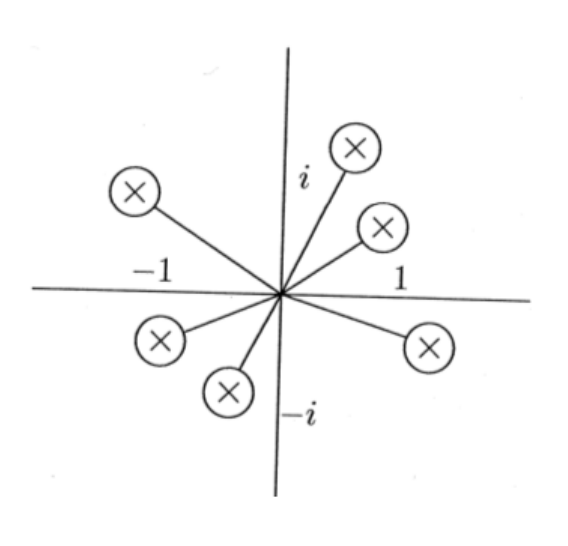}
				\caption{An example of paths  in the $\la$-plane chosen to determine the monodromy for pairs  $(A,B)$ from $GE_3^\bC$.}
			 \label{Fig7}
			\end{figure}\end{center}
			
			\begin{center}
		\begin{figure}[h]		
				\includegraphics[scale=0.4]{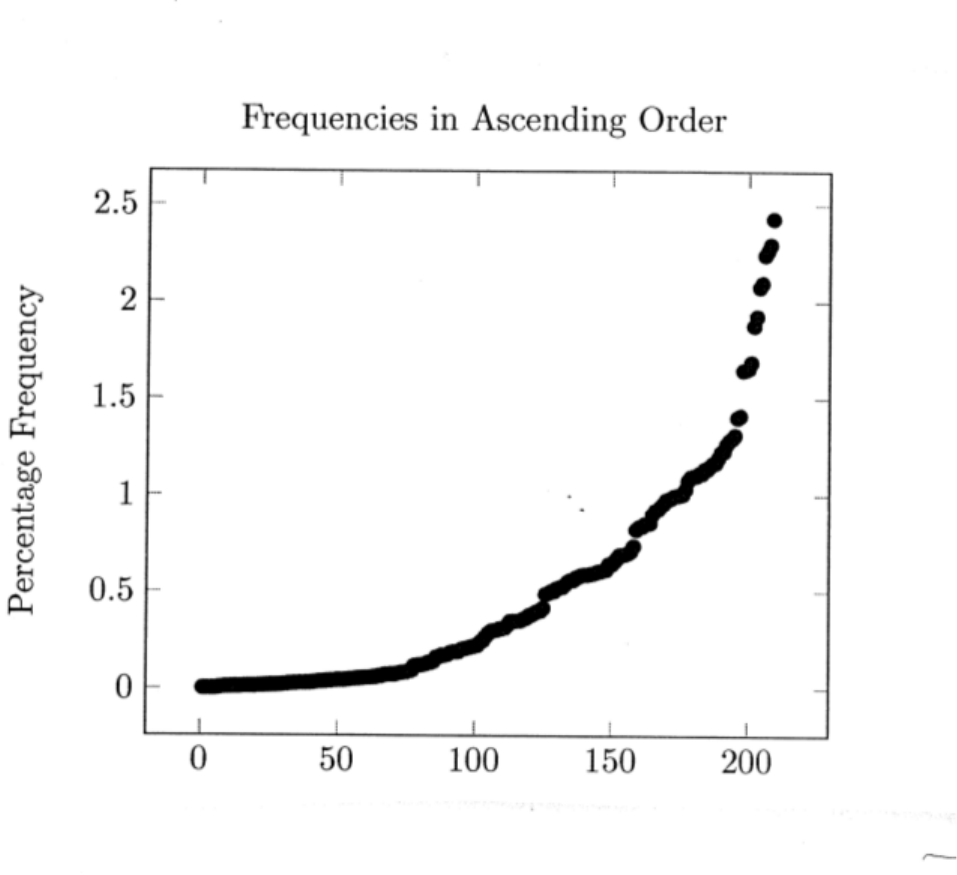}
				\caption{Frequencies of 240 possible $6$-tuples of transpositions from $S_3$ in the ascending order.}
			 \label{Fig8} 
			\end{figure} \end{center}
		
\medskip
For $A$ and $B$  in $GE_3^\bC$, there are 240  sequences of $6$-tuples of  transpositions $(\si_1,\si_2,\dots , \si_6)$ from $S_3$ satisfying the conditions:

\noindent
(i) they generate the symmetric group $S_3$; 

\noindent
(ii) the product $\si_1 \cdot \si_2\cdot \ldots \cdot \si_6$ coincides with the identity permutation  $(1,2,3).$

\medskip		
Using a similar MATLAB-code to determine the monodromy transpositions, we  generated 150000 random matrix pairs in $GE_3^\bC$ and calculated their monodromy sequences. Our numerical results show the following, see Fig.~\ref{Fig8}. 

\medskip
\noindent (i) Of the 240 possible cases, only 209 were realized and only 204 were realized more	
		than once. 
		
\noindent (ii)		The most common monodromy sequences were  $(2 3)(1 2)(2 3)(1 2)(2 3)(1 2)$, which occurred with the frequency  2.43 \%  and  $(1 2)(1 3)(1 3)(2 3)(2 3)(1 2)$ which occurred with the frequency  2.29 \%. 

\noindent (iii) 
Monodromy sequences in which one permutation occurs four times in a row followed by two occurrences of another permutation and their cyclic permutations (for example, $(1 2)(1 2)(1 2)(1 2)(1 3)(1 3)$ or $(1 2)(2 3)(2 3)(2 3)(2 3)(1 2))$ were the most rare, occurring only once or not at all. 


\begin{remark}	\label{remfin}	One particularly strange and interesting result is that the labelling of the eigenvalues seems to affect the frequencies with which certain monodromy sequences appear. In the case of $GE_3^\bC$-matrices, one can  relabel the three preimages of $\la=0$, i.e., the eigenvalues of $A$,  by using the action of $S_3$. Usually, about half of these six group elements change the frequency by either doubling or halving the original one. The other half of the group tends to keep the frequency the same, but exactly which members of $S_3$ do what varies from case to case. We have not been able to find a pattern of or an explanation  to why relabelling changes the frequencies in this peculiar way.
\end{remark}

\section{Final remarks}\label{sec:final}

In connection with our topic, one can naturally ask why we only restrict ourselves to consideration of the distributions of a single level crossing point on $\bC$ and are not trying to obtain  information about the joint distribution of all $n(n-1)$ level crossing points which obviously exists in all the above cases.  It turns out that for $n>3$, not all  $n(n-1)$-tuples of complex numbers can be realized as  level crossings and even the description of the loci of realizable $n(n-1)$-tuples is very complicated. This fact definitely means that at least for $n>3$, to get the joint distribution of level crossings on such loci will be a formidable (if not completely impossible) task,  comp. e.g. \cite{OnSh}. On the other hand, in the simplest case $n=2$, we calculate and use such joint distributions below.  



\end{document}